\documentclass[12pt,draftclsnofoot,onecolumn]{IEEEtran}
\usepackage{amssymb}
\usepackage{amsmath}
\UseRawInputEncoding
\usepackage{graphicx}
\usepackage{cite}
\usepackage{color}
\usepackage{citesort}
\usepackage{multicol}
\usepackage{amsfonts}
\usepackage{bbm}
\usepackage{subfigure}
\usepackage{cite,amsmath,amssymb}
\usepackage{comment}
\usepackage{amsmath}
\usepackage{lipsum}
\usepackage{stfloats}
\usepackage{algorithm}
\usepackage{algorithmic}
\usepackage{makecell}
\usepackage{booktabs}
\usepackage{multirow}

\newlength{\halfpagewidth}
\divide\halfpagewidth by 2

\newtheorem{theorem}{\textbf{Theorem}}

\newtheorem{lemma}{\textbf{Lemma}}

\newtheorem{corollary}{\textbf{Corollary}}

\newtheorem{proof}{\textbf{Proof}}

\makeatletter
\def\ScaleIfNeeded{%
\ifdim\Gin@nat@width>\linewidth \linewidth \else \Gin@nat@width
\fi } \makeatother

\begin{document}
%

\title{\huge{V2I-Based Platooning Design with Delay Awareness}}
\author{Lifeng Wang, Yu Duan, Yun Lai, Shizhuo Mu, and Xiang Li
\thanks{Authors are with the Department of Electrical Engineering, Fudan University, Shanghai, China (E-mail: $\rm\{lifeng.wang,lix\}@fudan.edu.cn$).}
}

\maketitle

\begin{abstract}
This paper studies the vehicle platooning system based on vehicle-to-infrastructure (V2I) communication, where all the vehicles in the platoon upload their driving state information to the roadside unit (RSU), and RSU makes the platoon control decisions with the assistance of edge computing. By addressing the delay concern, a platoon control approach is proposed to achieve plant stability and string stability. The effects of the time headway, communication and edge computing delays on the stability are quantified. The velocity and size of the stable platoon are calculated, which show the impacts of the radio parameters such as massive MIMO antennas and frequency band on the platoon configuration. The handover performance between RSUs in the V2I-based platooning system is quantified by considering the effects of the RSU's coverage and platoon size, which demonstrates that the velocity of a stable platoon should be appropriately chosen, in order to meet the V2I's Quality-of-Service and handover constraints.
\end{abstract}


\section{Introduction}
The commercially-used adaptive cruise control (ACC) enables vehicles to maintain safe inter-vehicle distance, which can avoid the collision and achieve autonomous driving through following the vehicle ahead~\cite{C_Y_Liang_1999}. To obtain the inter-vehicle distance and relative velocity, such an intelligent transportation system (ITS) fully depends on the vehicle's radar sensing capability~\cite{Ramzi2003}. However, the drawback of radar sensor is that its efficacy could be degraded by the obstructions or bad weather. More importantly, ACC system is
susceptible to the string instability, which results in phantom traffic jams~\cite{Daniel_Work2020}. Cooperative adaptive cruise control (CACC) is a promising approach to deal with these issues~\cite{X_Liu_2001,Soncu_2014}. As the extension of ACC, CACC allows vehicles to communicate with each other for sharing their driving state information (DSI) such as position, spacing, velocity, acceleration/deceleration rate, and time headway etc. Compared to the ACC, CACC can provide earlier collision avoidance, traffic jam mitigation, aerodynamic drag force reduction, and extended sensors~\cite{G_Naus_2010,J_Ploeg_2011,Vicente2014,3GPP_TS_V2X}.

Vehicle-to-everything (V2X) communications allow vehicle-to-vehicle (V2V) or vehicle-to-infrastructure (V2I) connectivity in the CACC systems. The V2V-based CACC systems have been widely studied in the literature~\cite{Shaw2007,YuYu_Lin_2017,B_Liu_2017,Shengbo_ITSC_2018,S_Darbha_2019,Yuanheng_Zhu_2019,Tengchan_2019,Michal_Sybis_2019}. These works have shown that V2V communications improve the stability and reduce the time headway in ITS systems, which means that higher traffic throughput and fuel efficiency can be achieved. However, the connectivity configurations in CACC systems are various (See Fig. 10 in \cite{Linjun_2016} and Fig. 2 in \cite{Shengbo_2019_mag}), which may result in high complexity of the control design. Existing research contributions have pointed out that connections between the leader vehicle and following vehicles may be more critical in the platooning system~\cite{B_Liu_2017}, which is the typical CACC scenario. Moreover, more V2V connections in the CACC systems may not necessarily improve the robustness if the control gains are inappropriately selected~\cite{Linjun_2016}. The V2V transmission rate needs to be large enough, in order to mitigate the detrimental effects of communication delay on the system stability~\cite{X_Liu_2001,Mario_2015,HaitaoXing_2020}. In addition, interference in the V2V-based CACC systems could deteriorate the V2V's Quality-of-Service (QoS) and should be properly managed~\cite{YuYu_Lin_2017,Tengchan_2019}, however, such interference management problem is challenging in practical dense traffic scenario~\cite{A_He_2017}. The V2I-based ITS systems have also attracted much attention~\cite{Vicente_2012,LeiChen_2016,yuyu_lin_ITA,Montanaro_2018,Vilalta_2019,Chang_2020}. It is known that V2I provides high-reliable and low-latency communication compared to the V2V, and the edge and central cloud computing resources~\cite{xiaoyanhu_2020} can be utilized in the V2I-based ITS systems. Therefore, V2I can ensure that the traffic flow is managed more efficiently and message dissemination is cost-effective~\cite{Vicente_2012,Montanaro_2018,Chang_2020}, particularly in dense traffic scenario with multi-platoons~\cite{yuyu_lin_ITA}.

In the CACC systems, vehicle platooning enables following vehicles to autonomously reach the leader vehicle's moving speed and keep the desired inter-vehicle distance while guaranteeing the safety and stability. Such maneuver control functionality can improve the road throughput and disengage the following vehicles from driving tasks. The aforementioned works mainly focus on the V2V-based  platooning systems. Due to its distributed feature, following vehicles undergo different levels of communication delays and different numbers of V2V links in the V2V-based platooning systems, which makes the platooning design challenging~\cite{X_Liu_2001,Mario_2015,LX_2011}. The V2V-based platooning also has to bear the extra burden of the heterogeneous control mechanisms and hardware resulted from different types of vehicles. To address these issues, this paper proposes an V2I-based platooning design. Compared to the conventional vehicle platooning systems with V2V communications, the advantages of the proposed design are: i) The majority of existing vehicle platooning schemes highly depend on the V2V links, which cannot support long-range communications and are subject to the blockages and severe interference in the dense traffic scenarios. Moreover, following vehicles that cannot directly communicate with the leader vehicle or other vehicles have to let other vehicles relay the vehicles' DSI, which makes the reliability compromised and inevitably results in high-latency. The proposed design only requires V2I connections, which are usually line-of-sight (the roadside units (RSUs) could be sites on the lamp posts); ii) By putting the platoon controller at the RSU with edge computing capability, the proposed design disengages following vehicles from making maneuver control decisions and enables simultaneous maneuver among vehicles in a platoon through sending control commands to the vehicles' actuators at the same time, in contrast to the V2V-based designs that different vehicles receive vehicles' DSI and carry out control decisions at the different time; iii) In existing platoon systems, any changes involving targeted inter-vehicle distance and vehicle's velocity have to be known by all the vehicles in a platoon, in order to change their states for new formation. In the proposed design, such changes only need to be known at the RSU, which will update the control commands accordingly. Therefore, the proposed design is more efficient and scalable for platoon management.

The main contributions of this paper are concluded as follows:
\begin{itemize}
\item \textbf{V2I-based Platooning Control Design:} In the considered system, all the vehicles' DSI are uploaded to the RSU via massive multiple-input multiple-output (MIMO), and RSU makes the platooning control decisions including the targeted velocity of the platoon based on the proposed control design. After computing the control inputs of all the following vehicles, RSU sends them to the following vehicles at the same time and frequency band.
\item \textbf{Plant and String Stability for the Proposed Platooning Solution:} In light of the communication and computing delay concern, the feasible control gain regions for meeting the plant stability and string stability are presented, respectively. We show that the control gains of the proposed platooning solution can be easily determined by using the D-subdivision method, in order to achieve plant stability. The effects of time headway on the stability are quantified.
\item \textbf{Relationships between Platoon's Velocity, Radio Parameters and Handover:} To achieve the required QoS of the V2I and avoid frequent handover, the platoon's  velocity needs to be appropriately chosen. With the assistance of massive MIMO, we provide a tractable approach to explicitly quantify the relationships between platoon's velocity, handover and radio parameters including the number of massive MIMO antennas and frequency band.  A simple solution with the help of dual connectivity has been proposed to achieve the seamless platooning control when handover occurs. The results are useful guidelines for fast radio resource allocation and handover management.
\item \textbf{Design Insights:} Our results show that different control gains have a big impact on the time of reaching the system stability. Different external disturbances and delays give rise to dramatic variations in the vehicles' traveling speeds and spacing errors, but have negligible effect on the disturbance time period before reaching the system stability. The effect of platoon size on the platooning stability and efficiency is marginal, which confirms the scalability of the proposed design.
\end{itemize}

The rest of this paper is organized as follows. In Section~\ref{System_description}, the considered system model is described and
the platooning control design is proposed. The stability of the proposed control design is analyzed in Section~\ref{stability_section}.
The platoon's velocity and handover are determined in Section \ref{sec:velocity_handover}.
Section~\ref{sec:simulation} provides the simulation results. Finally, some concluding remarks are presented in Section~\ref{conclusion_section}.

\section{System Descriptions}\label{System_description}
\begin{figure}[t!]
\centering
\includegraphics[width=3.3 in]{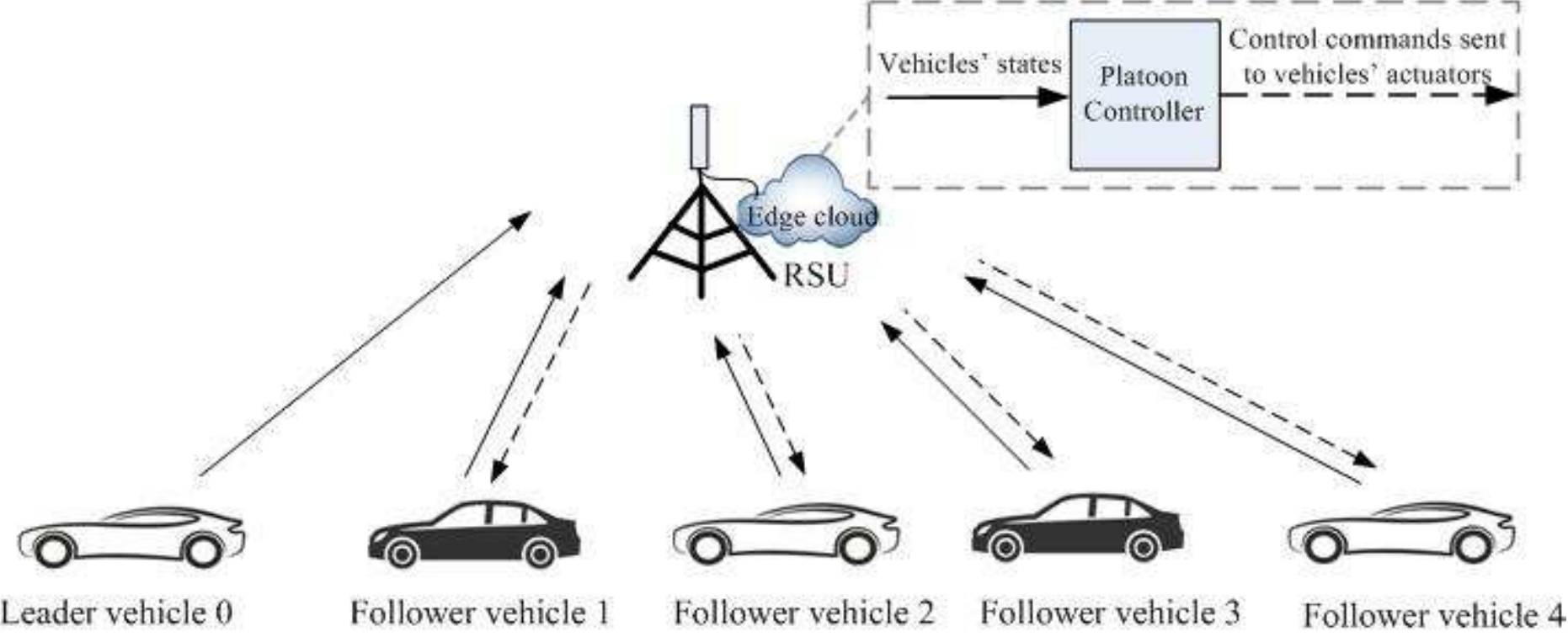}
\caption{An illustration of V2I-based platooning system with edge computing.
}
\label{edge_platoon}
\end{figure}
As illustrated in Fig.~\ref{edge_platoon}, we consider an V2I-based platooning system with massive MIMO, where each RSU equipped with $N$ antennas has edge computing capability~\cite{xiaoyanhu_2020}, and there are $M+1$ single-antenna vehicles in a platoon with the leader vehicle $0$ and follower vehicle $i$ ($i=1,\cdots,M$). In such a system, each vehicle simultaneously sends its DSI involving position and moving speed to the RSU\footnote{Note that vehicles' positions could be evaluated at RSU by applying positioning techniques~\cite{del_Peral-Rosado,Wymeersch_2017}, in this case, delay will be further cut because of less DSI uploaded to the RSU.}, which shall be processed by RSU for determining platooning control decisions. After edge cloud processing, RSU sends the control commands (i.e., desired acceleration values) to the corresponding follower vehicles' actuators. A point-mass model is considered to describe the longitudinal vehicle dynamics, which is given by~\cite{Linjun_2016,Mario_2015,Jianglin2020}
\begin{align}\label{dyna_model}
\dot{x}_i(t)=v_i(t),\;\; \dot{v}_i(t)=u_i(t),
\end{align}
where ${x}_i(t)$, $v_i(t)$, and $u_i(t)$ are the position, velocity, and control input (acceleration) of the vehicle $i$ at time $t$, respectively. The spacing error is defined as
\begin{align}\label{spacing_error}
e_i(t)=x_i \left( t \right) - x_{i - 1} \left( t \right) + h v_o+ l,
\end{align}
where $h$ is the time headway, $v_o$ is the targeted platoon's velocity (The selection of $v_o$ value will be illustrated in Section \ref{sec:velocity_handover}), and $l$ is standstill distance, $h v_o+ l$ is the desired inter-vehicle distance. As the leader vehicle travels at the constant speed of $v_o$, the platooning rule is
\begin{align}\label{platoon_rule}
\mathop {\lim }\limits_{t \to \infty } e_i \left( t \right) = 0,\;\;\mathop {\lim }\limits_{t \to \infty } v_i \left( t \right) = v_o.
\end{align}

Since all the vehicles undergo the identical communication delay and the processing delay with the assistance of massive MIMO and edge computing, the platooning control law at the RSU is designed as
\begin{align}\label{control_law}
u_i(t)&=  - K_x \left(x_i \left( t - \tau  \right) - x_{i - 1} \left( t - \tau  \right) + hv_i \left(t - \tau \right)+l\right)\nonumber\\
 &- K_v \left( v_i \left(t - \tau\right) - v_{i-1} \left(t - \tau\right)  \right) - K_{v_o } \left(v_i \left(t - \tau \right) - v_o  \right) \nonumber\\
&~- K_{x_o } \left(x_i \left(t - \tau \right) - x_o \left( t - \tau \right) + ihv_o + i l\right),
\end{align}
where $K_x$, $K_v$, $K_{v_o }$, and $K_{x_o}$ are positive control gains, $\tau$ is the total amount of the delay resulted from the communication and edge cloud processing.  To guarantee the platoon stability, the control gains need to be chosen appropriately.

\noindent \textbf{Remark 1:} The proposed control law only utilizes the DSI of the leader vehicle and the follower vehicle $i$ for determining the vehicle $i$'s control input. Although existing V2V-based platooning control designs~\cite{X_Liu_2001,Shaw2007,S_Darbha_2019,Tengchan_2019} have attempted to make the most of these DSI, the effects of time headway~\cite{X_Liu_2001} or communication delay~\cite{Shaw2007,S_Darbha_2019} may be ignored for tractability, or some quite conservative conditions are required~\cite{Tengchan_2019}. Another benefit of the proposed V2I-based platooning design is that the control gains for system stability can be easily calculated, which is illustrated in the next section.

\section{Stability Analysis}\label{stability_section}
In this section, the control gains in \eqref{control_law} are determined from the perspective of plant stability and string stability. To facilitate the stability analysis, a frequency-domain approach is adopted. According to \eqref{dyna_model},  we have
\begin{align}\label{error_control_input}
\ddot{x}_i \left( t \right)-\ddot{x}_{i-1} \left( t \right) = u_i(t)-u_{i-1}(t).
\end{align}
Substituting \eqref{control_law} into \eqref{error_control_input}, after mathematical manipulations, \eqref{error_control_input} is rewritten as
\begin{align}\label{error_control_Relationship}
&\ddot{x}_i \left( t \right)- \ddot{x}_{i-1} \left( t \right) = - \lambda \left(x_i \left( {t - \tau } \right) - x_{i - 1} \left( {t - \tau } \right)\right)  \nonumber\\
 &~~+K_x \left(x_{i - 1} \left(t - \tau \right) - x_{i - 2} \left(t - \tau\right)\right) \nonumber\\
&~~-\eta \left(v_{i} \left(t - \tau\right) - v_{i-1} \left(t - \tau\right)\right)   \nonumber\\
  & ~~ + K_v \left(v_{i-1} \left(t - \tau\right) - v_{i-2} \left(t - \tau\right)\right)-K_{x_o} \left(h v_o+ l\right),
\end{align}
where $\lambda = K_x+K_{x_o }$ and $\eta=K_x h+  K_v + K_{v_o }$. Let $E_i\left(s\right)=\mathcal{L}\left\{e_i \left(t\right)\right\}$ denote the Laplace transform of the spacing error $e_i \left(t\right)$, taking the Laplace transform of \eqref{spacing_error} yields
\begin{align}\label{error_laplace}
\mathcal{L}\left\{x_i \left(t - \tau \right)-x_{i-1} \left(t - \tau \right)\right\} = e^{-\tau s} E_i\left(s\right) - e^{-\tau s} \frac{h v_o+l}{s}.
\end{align}
Based on \eqref{error_laplace},  the Laplace transform of \eqref{error_control_Relationship} is given by
\begin{align}\label{error_transfer_eq}
E_i\left(s\right) &= \frac{\left(K_v s + K_x\right)e^{-\tau s}}{\Theta\left(s\right)} E_{i-1}\left(s\right) \nonumber\\
&~~+ \frac{s+\left(\eta-K_v\right) e^{-\tau s} + \left(e^{-\tau s}-1\right)\frac{K_{x_o}}{s}}{\Theta\left(s\right)} \left(h v_o+ l\right),
\end{align}
where $\Theta\left(s\right)=s^2+\eta s e^{-\tau s} +\lambda e^{-\tau s}$ is referred to as characteristic function. Therefore, in the proposed platooning design,  the spacing error transfer function is calculated as
\begin{align}\label{error_transfer_eq}
\mathcal{H}_i\left(s\right) = \frac{\left(K_v s + K_x\right)e^{-\tau s}}{\Theta\left(s\right)}.
\end{align}

\subsection{Plant Stability}
Plant stability is achieved when the platooning rule given by \eqref{platoon_rule} is met. As such, the necessary and sufficient condition for satisfying the plant stability is
\begin{align}\label{condition_plant}
\mathrm{Re}\left(s_0\right) < 0,~\forall~\Theta\left(s_0\right) = 0,
\end{align}
which means that for an arbitrary characteristic root of $\Theta\left(s\right)$, it has negative real part. The complexity of solving \eqref{condition_plant} depends on the specific spacing error transfer function, which is determined by the platooning control law. The use of the Routh-Hurwitz criterion with Pad\'{e} approximation requires that the spacing error transfer function for the frequency range of interest can be well approximated~\cite{H_Xing_2016,Tengchan_2019,HaitaoXing_2020}, which may bring in more complexity. Considering the proposed platooning law given by \eqref{control_law}, we show that the control gains for achieving plant stability can be easily and precisely obtained by leveraging the D-subdivision method~\cite{D_subdivision_method}. Based on \eqref{condition_plant}, we have the following theorem:
\begin{theorem}
Plant stability can be guaranteed if and only if $\left(\lambda, \eta\right)$ belongs to the feasible region:
\begin{align}\label{w_eta_gamma_plane}
\mathcal{G}\left(\tau\right)=&\Bigg\{\left(\lambda, \eta \right): \lambda \leq  w^2 \cos \left( {\tau w} \right), \nonumber\\
&~~~~~~~~~~~~~~~~\eta \leq w\sin \left( {\tau w} \right), w \in\left(0, \frac{\pi}{2\tau}\right) \Bigg\}.
\end{align}
\end{theorem}
\begin{proof}
See Appendix A.
\end{proof}

\noindent \textbf{Remark 2:} As shown in Fig. \ref{fig_lamba_eta}, the size of the feasible region $\mathcal{G}\left(\tau\right)$ decreases as delay increases. Based on \textbf{Theorem 1}, we see that $\eta < \frac{\pi}{2\tau}$. Therefore, for a specific $\eta^* \in\left(0, \frac{\pi}{2\tau}\right)$, the critical value $w^*$ for $\eta^* = w^*\sin \left(\tau w^* \right)$ can be efficiently calculated by using one-dimension search since $w\sin \left(\tau w \right)$ is the increasing function of $w \in\left(0, \frac{\pi}{2\tau}\right) $. Then, we can obtain the corresponding $\lambda^* =  (w^*)^2 \cos \left( {\tau w^*} \right)$. In light of the point $\left(\lambda^*, \eta^* \right)$ on the D-curve (See Appendix A), the plant stability requires $\lambda \in \left(0,\lambda^*\right)$  for a specific $\eta^* \in\left(0, \frac{\pi}{2\tau}\right)$.
\begin{figure}[htbp]
\centering
\includegraphics[width=3.5 in,height=2.9 in]{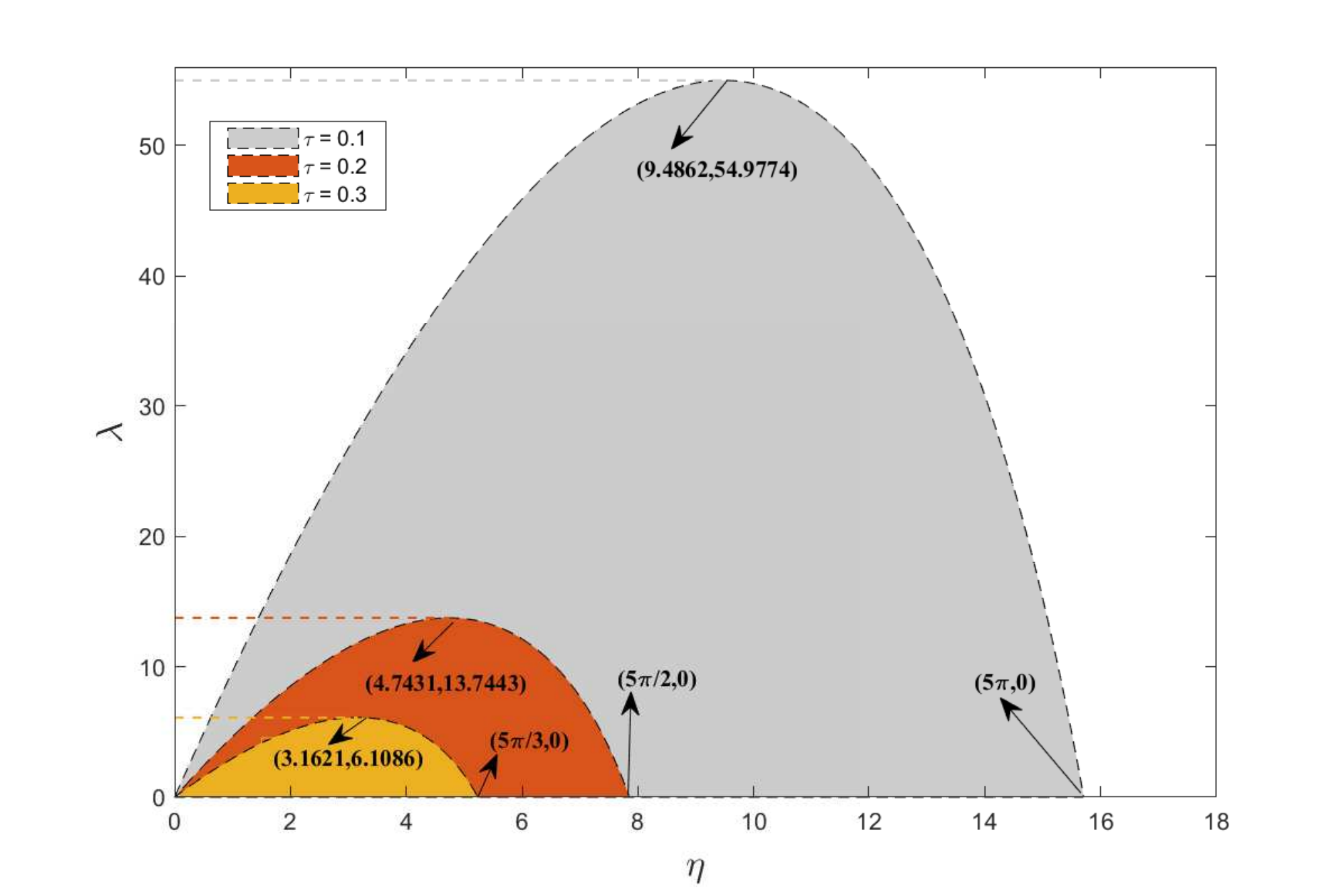}
\caption{The plant stability region $\mathcal{G}\left(\tau\right)$ for different levels of delay with different corner points.}
\label{fig_lamba_eta}
\end{figure}

\subsection{String Stability} In the platooning systems, unstable vehicle strings give rise to phantom traffic jams~\cite{Daniel_Work2020}. String stability ensures that the spacing error is not amplified in the traffic flow upstream~\cite{S_Darbha_2019,HaitaoXing_2020}, namely the magnitude of the spacing error transfer function $\mathcal{H}_i\left(s\right)$ needs to satisfy $\left|\mathcal{H}_i\left(jw\right)\right| < 1$. As such, we have the following theorem:
\begin{theorem}
String stability can be guaranteed when $\left(\lambda, \eta\right)$ belongs to the feasible region:
\begin{align}\label{string_stability}
\mathcal{S}\left(\tau\right) = \Bigg\{\left(\lambda, \eta \right): \lambda \leq  K_v K_{v_o }, \eta \le \frac{1}{2\tau}\Bigg\}.
\end{align}
\end{theorem}
\begin{proof}
See Appendix B.
\end{proof}

\noindent\textbf{Remark 3:} From \eqref{string_stability}, we see that the size of the feasible region $\mathcal{S}\left(\tau\right)$ decreases as delay increases. The time headway satisfies $h<\left(\frac{1}{2\tau}-K_v - K_{v_o }\right)/K_x$. Compared to the platooning method of \cite{LX_2011} with ACC where the time headway has to be larger than $2\tau$ for string stability, our design can keep the time headway at a minimum required level by selecting the proper control gains based on \eqref{string_stability}, hence the road throughput can be significantly improved.

\section{Platoon's Velocity and Handover} \label{sec:velocity_handover}
The previous section has provided the stability regions of the proposed platooning design given a targeted velocity of the stable platoon. In practice, the targeted velocity of a stable platoon has to be chosen appropriately, which has a big impact on the inter-vehicle distance, platoon size/length and QoS of the V2X communications. Unfortunately, such concern has not been paid enough attention yet. Existing works such as \cite{YuYu_Lin_2017} have shown that
inappropriate inter-vehicle distance in a platoon could deteriorate the message dissemination in the V2V links. Research efforts have focused on how to obtain the optimal inter-vehicle distance under QoS constraint~\cite{yuyu_lin_ITA}. However, the study of the relationships between platoon's velocity, time headway, handover and radio parameters is still in its infancy. Some critical concerns in the early works such as massive information exchange for centralized formation control~\cite{Shaw2007} can be easily addressed now, since the radio technologies have developed faster than ever before. In this section, we seek a low-complexity approach to answer the following questions:
\begin{itemize}
  \item How to quantify the relationship between the RSU coverage and platoon size/length?
  \item How to allocate the radio resources given a platoon configuration?
  \item How to manage the handover between RSUs given a platoon configuration?
\end{itemize}
 It is de facto challenging to find a generic solution for these questions. As such, we consider the platooning systems with the massive MIMO aided V2I communications. Massive MIMO is one of key 5G radio technologies and enables communications with dozens of users at the same time and frequency band~\cite{Marzetta_2010_Nonc}. Moreover, it can achieve high-speed transmission rate, combat the co-channel interference, and facilitate resource allocation~\cite{Bjorson_mag_2016,xiaoyanhu_2020}.

We adopt a linear massive MIMO processing method for V2I communication, i.e., zero-forcing (ZF) detection is implemented at RSU. The achievable communication rate (bps) of the vehicle $i$ is given by~\cite{ngo2013energy}
\begin{align}\label{capacity_rate}
R_i= B \log _2 \left( {1 + \frac{{P_{v_i} \left( {N - M-1} \right)\beta {d_i}^{ - \alpha} }}{{\sigma ^2 }}} \right),
\end{align}
where $B$ is the platoon system bandwidth, $P_{v_i}$ is the vehicle $i$'s transmit power, $\beta$ is the constant parameter commonly-set as ${(\frac{{\text{c}}}{{4\pi {f_c}}})^2}$ with $c=3 \times 10^8 \rm m/s$ and the carrier frequency $f_c$, $d_i$ is the communication distance, $\alpha$ is the path loss exponent, and $\sigma ^2$ is the noise power. Note that due to the ``channel hardening'' feature of massive MIMO~\cite{Bjorson_mag_2016,xiaoyanhu_2020}, the small-scale fading effects are averaged out. Therefore, given a minimum communication rate threshold $R_\mathrm{th}$ (namely QoS constraint), the radius of the RSU coverage is
\begin{align}\label{RSU_range}
d_{\rm th}=\left( {\frac{{P_{v_i } \left( {N - M - 1} \right)\beta }}{{\sigma ^2 \left( {2^{\frac{{R_{\rm th} }}{B}}  - 1} \right)}}} \right)^{1/\alpha }.
\end{align}
Let $r_o$ and $h_o$ denote the perpendicular distance and the absolute antenna elevation difference between the platoon vehicle and the RSU, respectively, based on \eqref{RSU_range}, the maximum longitudinal coverage range of the RSU is
\begin{align}\label{longitudinal_distance}
\ell_\mathrm{th} = \sqrt {d_{\rm th}^2  - r_o^2  - h_o^2 }.
\end{align}

For a specific targeted velocity of the stable platoon $v_o$, the platoon size/length is calculated as
\begin{align}\label{platoon_size}
\mathcal{D}_\mathrm{platoon} = M h v_o + M l.
\end{align}
The traveling time for a stable platoon in an RSU coverage area before undergoing handover is
\begin{align}\label{time}
T_\mathrm{stay}= \frac{2\ell_\mathrm{th}^0 -\mathcal{D}_\mathrm{platoon}}{v_o},
\end{align}
where $\ell_\mathrm{th}^0$ is calculated by using \eqref{longitudinal_distance} with $P_{v_i} =P_{v_0}$, due to the fact that the leader vehicle is the first to leave an RSU's coverage area.  Let $f_\mathrm{handover}$ denote the maximum allowable handover frequency between RSUs, in other words, the minimum traveling duration for a platoon in an RSU coverage area is $1/f_\mathrm{handover}$. It is obvious that $T_\mathrm{stay}$ should be greater than $1/f_\mathrm{handover}$. Thus,  by considering  \eqref{platoon_size} and \eqref{time}, we have the following condition:
\begin{align}\label{platoonvelocity}
v_o \leq \frac{2\ell_\mathrm{th}^0 - M l}{M h + 1/f_\mathrm{handover}}.
\end{align}

\noindent \textbf{Remark 4:} It is indicated from \eqref{platoonvelocity} that given the radio resources and handover frequency, platoon's velocity decreases when time headway increases, i.e., there is a tradeoff between platoon's velocity and time headway. Given a platoon configuration, the minimum required number of massive MIMO antennas or bandwidth under the QoS and handover constraints can
 \begin{table}[h]\label{table1}
\centering
\caption{Results based on \eqref{platoonvelocity}}
\begin{tabular}{|c|ccc|}
\hline
$f_c$            & $R_\mathrm{th}$(Mbps)  & $f_\mathrm{handover}$(times/s)    & maximum $v_o$(m/s)  \\ \hline
\multirow{3}{*}{3.5GHz} & 75 & 1/30 & 24 \\
                     & 75 & 1/20 & 35 \\
                     & 75 & 1/10 & 65 \\ \hline
\multirow{3}{*}{5.9GHz} & 75 & 1/30 & 14 \\
                     & 75 & 1/20 & 20 \\
                     & 75 & 1/10 & 38 \\ \hline
\multicolumn{4}{l}{N.B.: In the table, $\tau=0.3$s, $h=0.2$s, $N=64$, $M=9$, $ML=15$m,}\\
\multicolumn{4}{l}{$r_o=10$m, $h_o=6$m, $\alpha=2$, $\beta = {(\frac{{\text{c}}}{{4\pi {f_c}}})^2}$, $B=5$MHz, }\\
\multicolumn{4}{l}{$P_{v_0}=20$dBm, $\sigma^2 =- 174 + 10\log _{10}\left(B\right)$dBm.}
\end{tabular}
\end{table}
be easily evaluated based on \eqref{platoonvelocity}. Therefore, \eqref{platoonvelocity} is useful for the fast radio resource allocation and handover management in the platoon systems.
 As shown in the Table~I, higher platoon's velocity results in more handovers for the same frequency band, and higher frequency band reduces the level of the maximum allowable platoon's velocity for a fixed number of antennas and bandwidth($N=64$ and $B=5$MHz in the Table I). To keep the desired levels of platoon's velocity and QoS, more numbers of antennas and bandwidths are demanded in the higher frequencies.

\begin{figure}[t!]
\centering
\includegraphics[width=3.2 in]{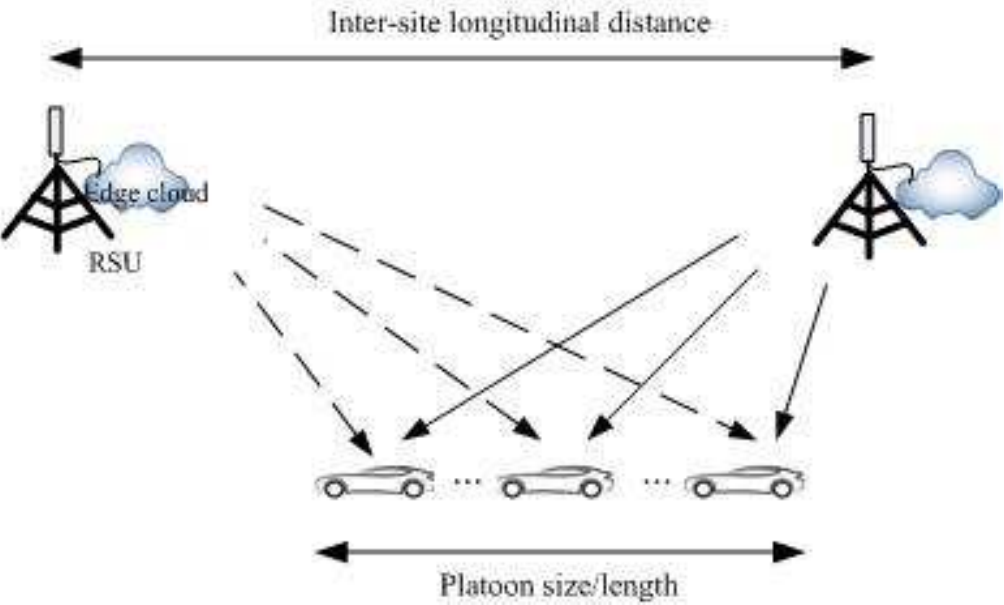}
\caption{A platoon can be seamlessly served by RSUs as the platoon is in the dual-connectivity range during the handover.
}
\label{platoon_handoverfig}
\end{figure}
The aforementioned has shown how to manage the platoon's velocity and radio resources in order to avoid frequent handover and meet the QoS requirement. In practice, it is important that the V2I-based platoon can be seamlessly controlled by RSUs when handover occurs. We realize that dual connectivity has been adopted in 4G and 5G systems~\cite{TS37_340,lifeng_mag2018}, to enhance the mobility robustness in cellular networks. Since dual connectivity allows a user to communicate with multiple network nodes at the same time, the QoS constraint can be guaranteed during the handover.
 As shown in Fig. \ref{platoon_handoverfig}, the inter-site longitudinal distance (ISLD) should be kept at a certain level to ensure that the platoon is in the dual-connectivity range during the handover. Based on \eqref{RSU_range} and \eqref{longitudinal_distance},  we can easily calculate the maximum allowable ISLD for dual connectivity as
\begin{align}\label{longitudinal_distance_DC_case}
&\ell_\mathrm{ISLD}^{\rm max} = 2 \ell_\mathrm{th}^0-\mathcal{D}_\mathrm{platoon} \nonumber\\
&=2 {\left(\left( {\frac{{P_{v_0 } \left( {N - M - 1} \right)\beta }}{{\sigma ^2 \left( {2^{\frac{{R_{\rm th} }}{B}}  - 1} \right)}}} \right)^{2/\alpha } - r_o^2  - h_o^2 \right)}^{1/2}-\mathcal{D}_\mathrm{platoon}.
\end{align}
From \eqref{longitudinal_distance_DC_case}, we see that by using dual connectivity, the V2I-based platooning systems can be seamlessly served by RSUs when the ISLD is below $\ell_\mathrm{ISLD}^{\rm max} $. It should be noted that such V2I-based platooning handover approach is flexible, for instance, by managing the radio resources such as transmit power and the number of massive MIMO antennas in \eqref{longitudinal_distance_DC_case},   ISLD can be easily tailored to meet various circumstances.

\section{Numerical Results}\label{sec:simulation}
In this section, numerical results are provided to demonstrate the efficiency of the proposed V2I-based platooning design and validate our analysis. In addition, the effects of different control gains, external disturbances, platoon sizes and delays on the performance are illustrated.
\subsection{Efficiency of the Proposed Platooning Design}
\begin{table}
\centering
\caption{Simulation parameters in Figs. \ref{efficiency} and \ref{instability}}
\begin{tabular}{cccccc}
\hline
Fig. & $\tau$ & $K_v$ & $K_{v_o }$ & $K_x$ & $K_{x_o}$ \\ \hline\hline
\ref{efficiency}(a)  & 0.1s    &  0.75 &  0.75 & 0.273  &  0.281 \\
\ref{efficiency}(b)  & 0.2s    &  0.75 & 0.75  & 0.213  & 0.297  \\
\ref{efficiency}(c)  & 0.3s    & 0.75  &  0.75 &  0.249 & 0.228  \\
\ref{instability}     & 0.3s   &  0.1 & 0.2  & 0.5  & 0.1\\ \hline
\end{tabular}
\end{table}
\begin{figure}
     \centering
    \subfigure[]{
         \centering
         \includegraphics[width=3.7 in,height=2.0 in]{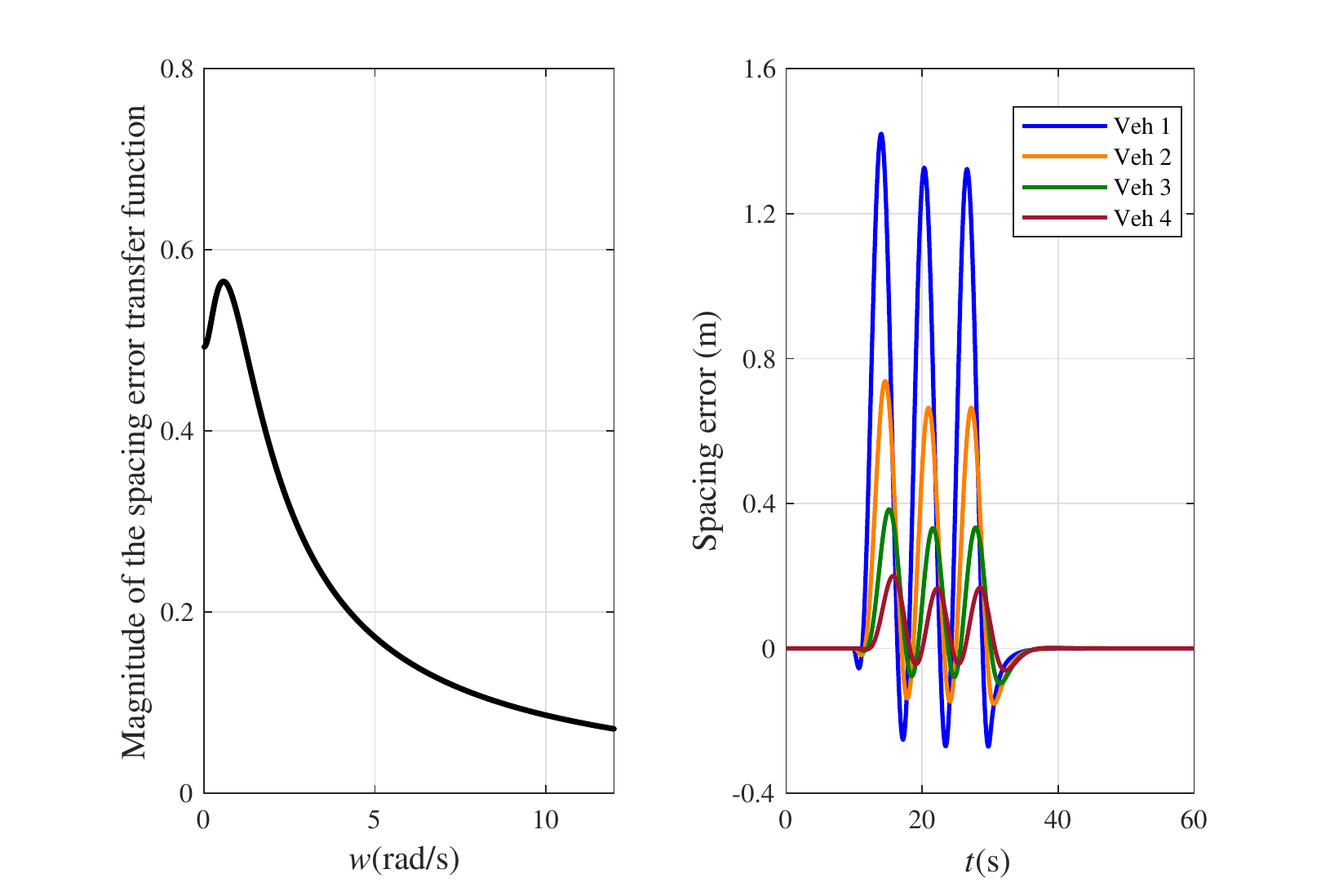}
      \label{fig1a}
     }
     \subfigure[]{
         \centering
         \includegraphics[width=3.7 in,height=2.0 in]{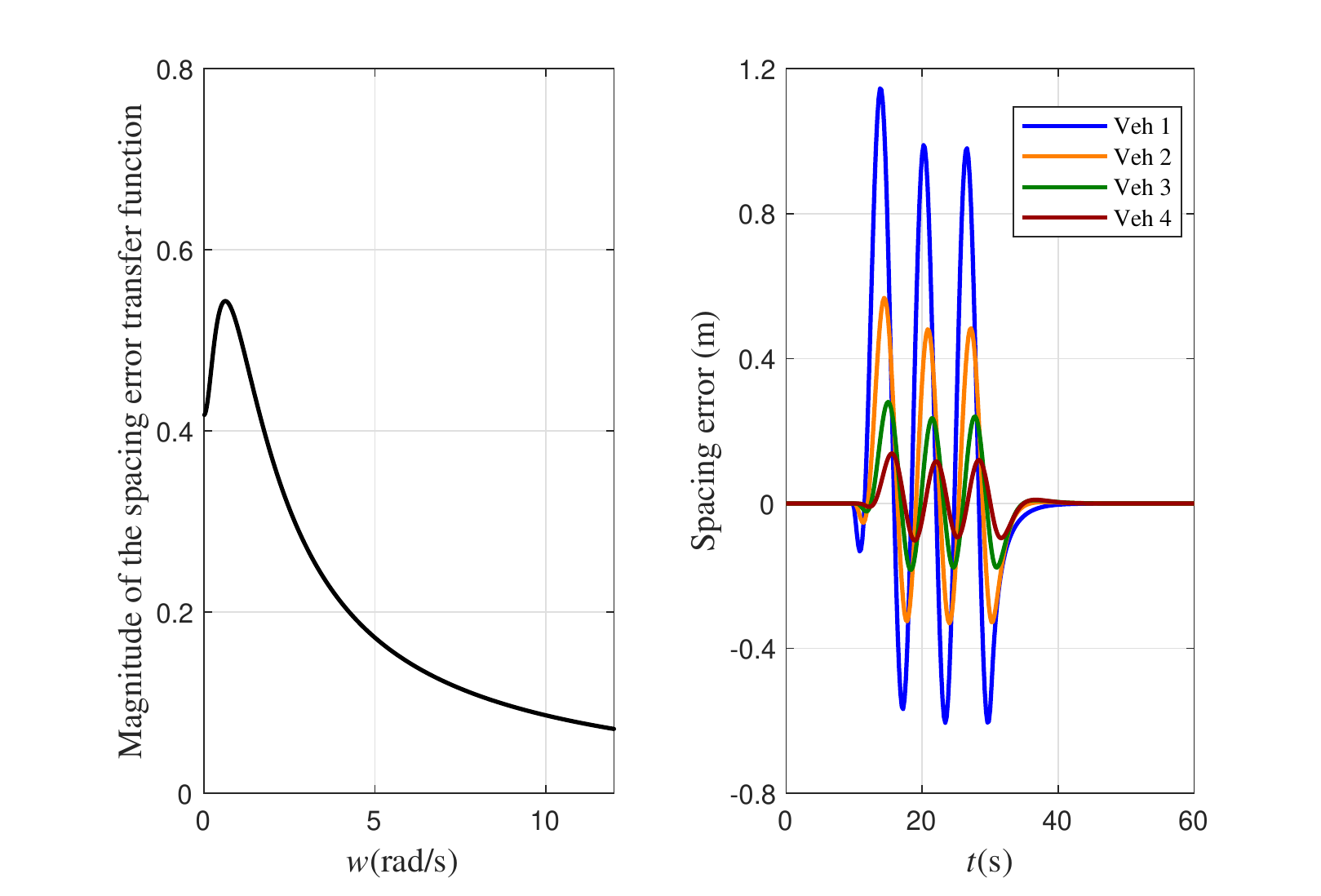}
      \label{fig1b}
    }
    \subfigure[]{
         \centering
         \includegraphics[width=3.7 in,height=2.0 in]{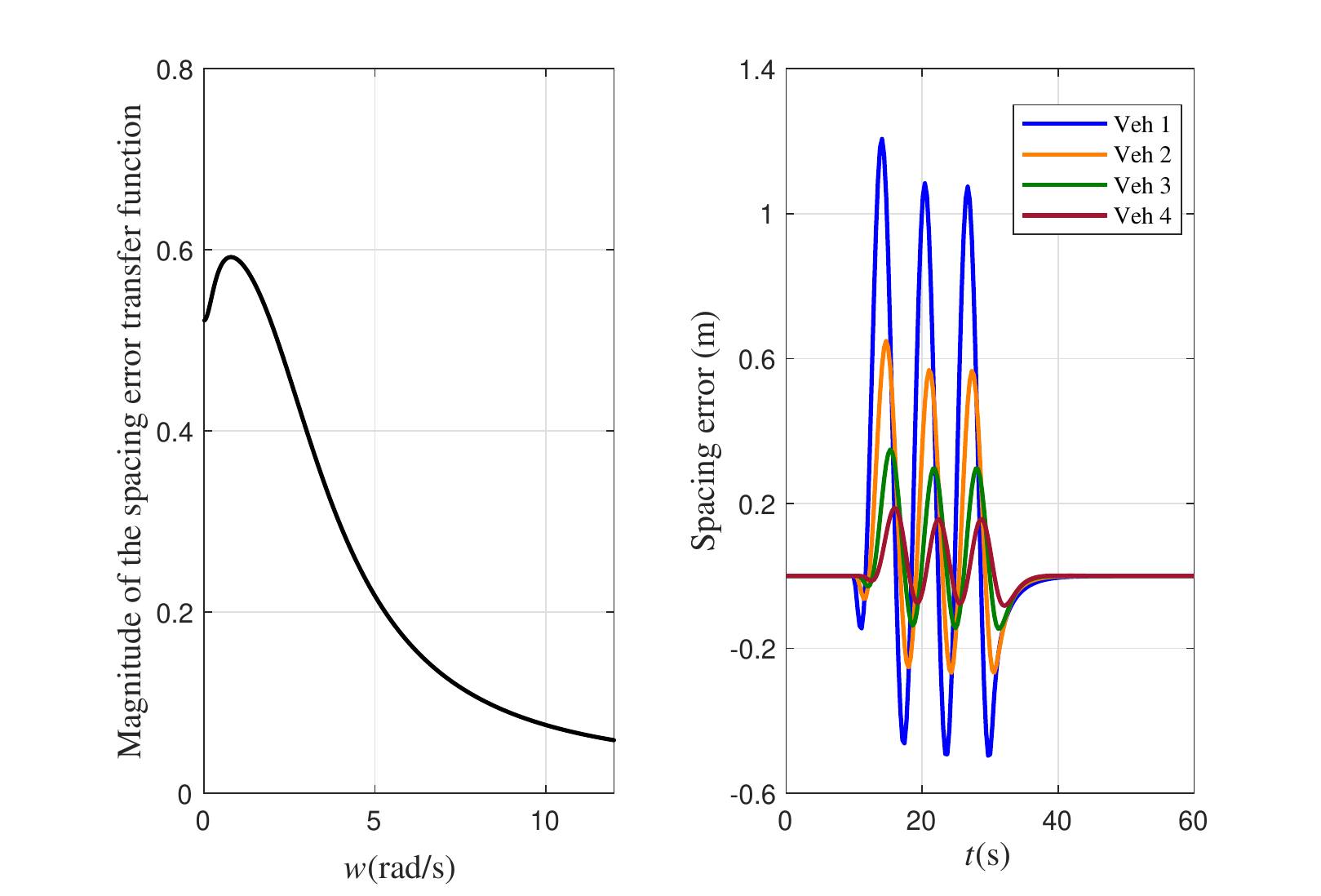}
         \label{fig1c}
     }
 \caption{The platooning performance of the proposed design.}
 \label{efficiency}
\end{figure}
This subsection shows the efficiency of the proposed design. In the simulations,  the time headway $h=0.2$s, the number of follower vehicles is $M=4$, and the spacing error is zero before leader vehicle changes its velocity.
The leader vehicle suffers an external disturbance during the time period $10\leq t\leq30$s, which is modeled by assuming that its acceleration varies as $\dot{v}_0(t)=-sin\left(t\right)$. The other system parameters are summarized in the Table II.

Fig. \ref{efficiency} shows the proposed platooning design can efficiently achieve plant stability and string stability for different levels of delay.  As mentioned in Theorem 2, the magnitude of the spacing error transfer function is kept below 1 for an arbitrary frequency $w$ and the spacing error decreases in the traffic flow upstream(namely the vehicle index increases) since the control gains are chosen from the feasible region $\mathcal{S}\left(\tau\right)$ given by \eqref{string_stability}. The spacing errors of the follower vehicles can be quickly diminished to zero when the leader vehicle's external disturbance is gone at $t>30$s, since the control gains belongs to the feasible region $\mathcal{G}\left(\tau\right)$ given by \eqref{w_eta_gamma_plane} and thus plant stability is guaranteed.

Fig. \ref{instability} shows the case when the control gains are chosen from the outside of $\mathcal{S}\left(\tau\right)$($\lambda >  K_v K_{v_o }$ in the Table II). As analyzed before, the magnitude of the spacing error transfer function is larger than 1 for certain $w$ values, in this case, the spacing errors of the follower vehicles are amplified in the traffic flow upstream, i.e., string instability occurs.
\begin{figure}[htbp]
\centering
\includegraphics[width=3.8 in,height=2.9 in]{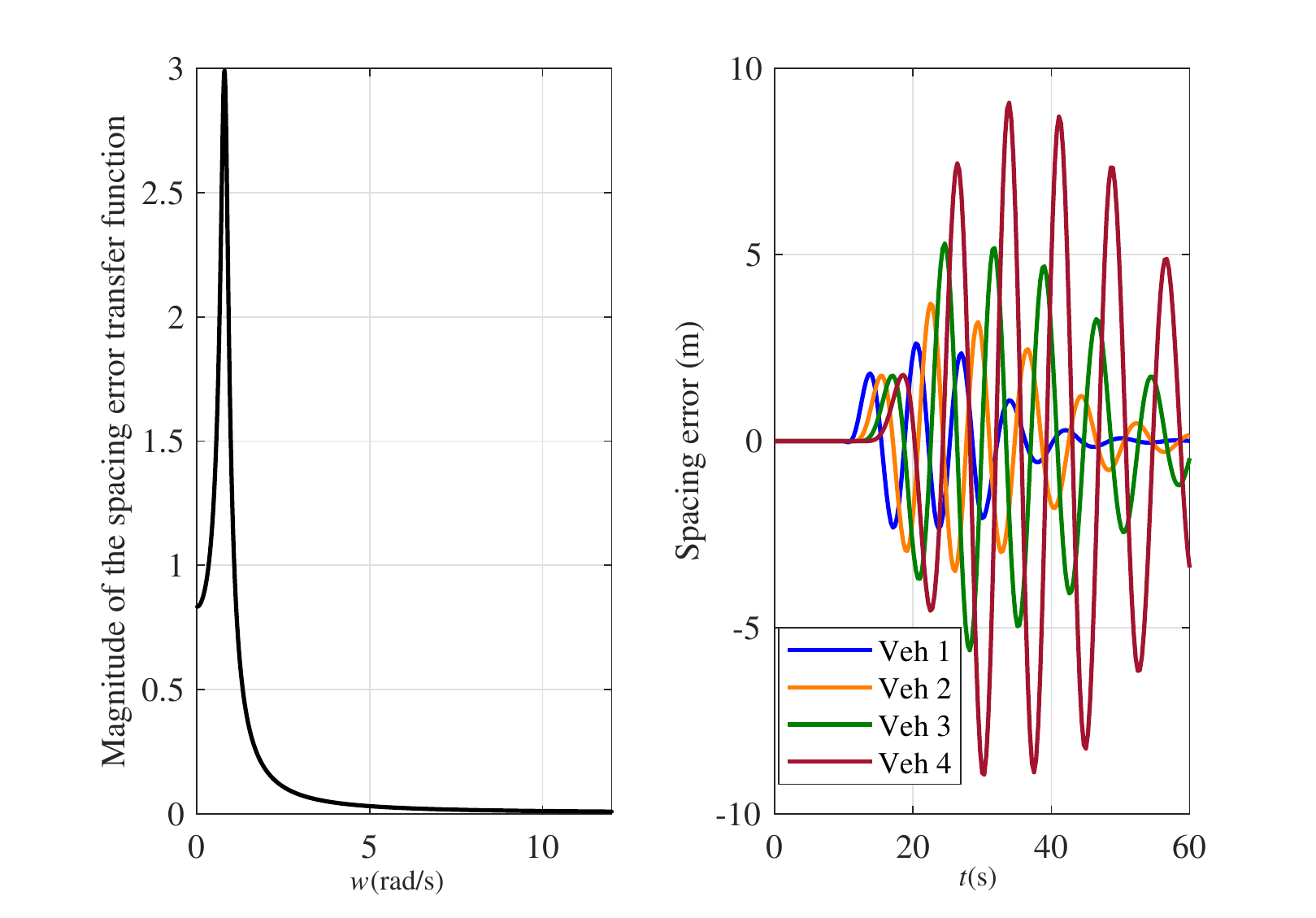}
\caption{String instability result when control gains do not belong to the feasible region given by \eqref{string_stability}.}
\label{instability}
\end{figure}

\subsection{Effects of Control Gains}

\begin{figure}
     \centering
    \subfigure[]{
         \centering
         \includegraphics[width=3.4 in,height=2.5 in]{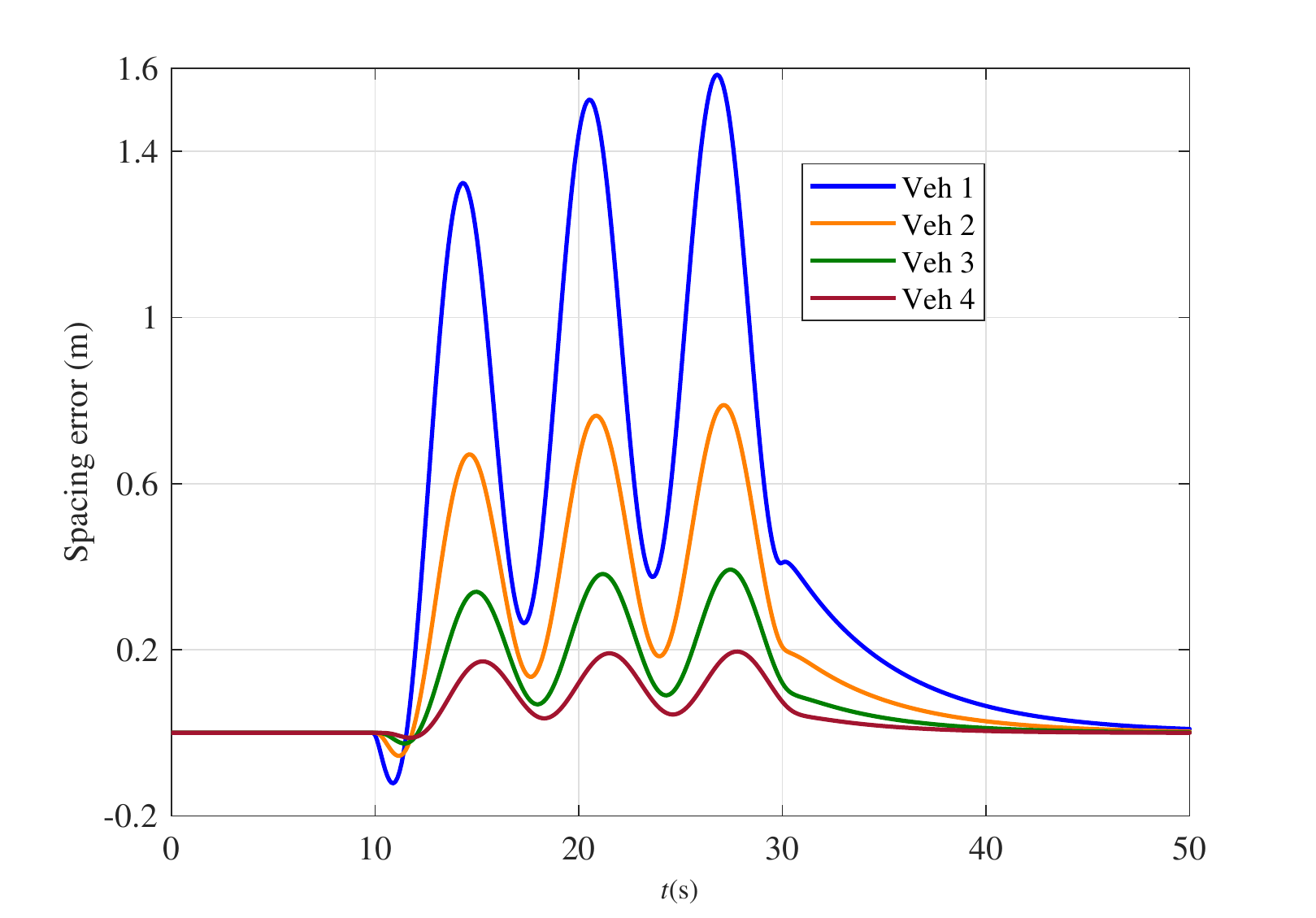}
      \label{fig1a}
     }
     \subfigure[]{
         \centering
         \includegraphics[width=3.4 in,height=2.5 in]{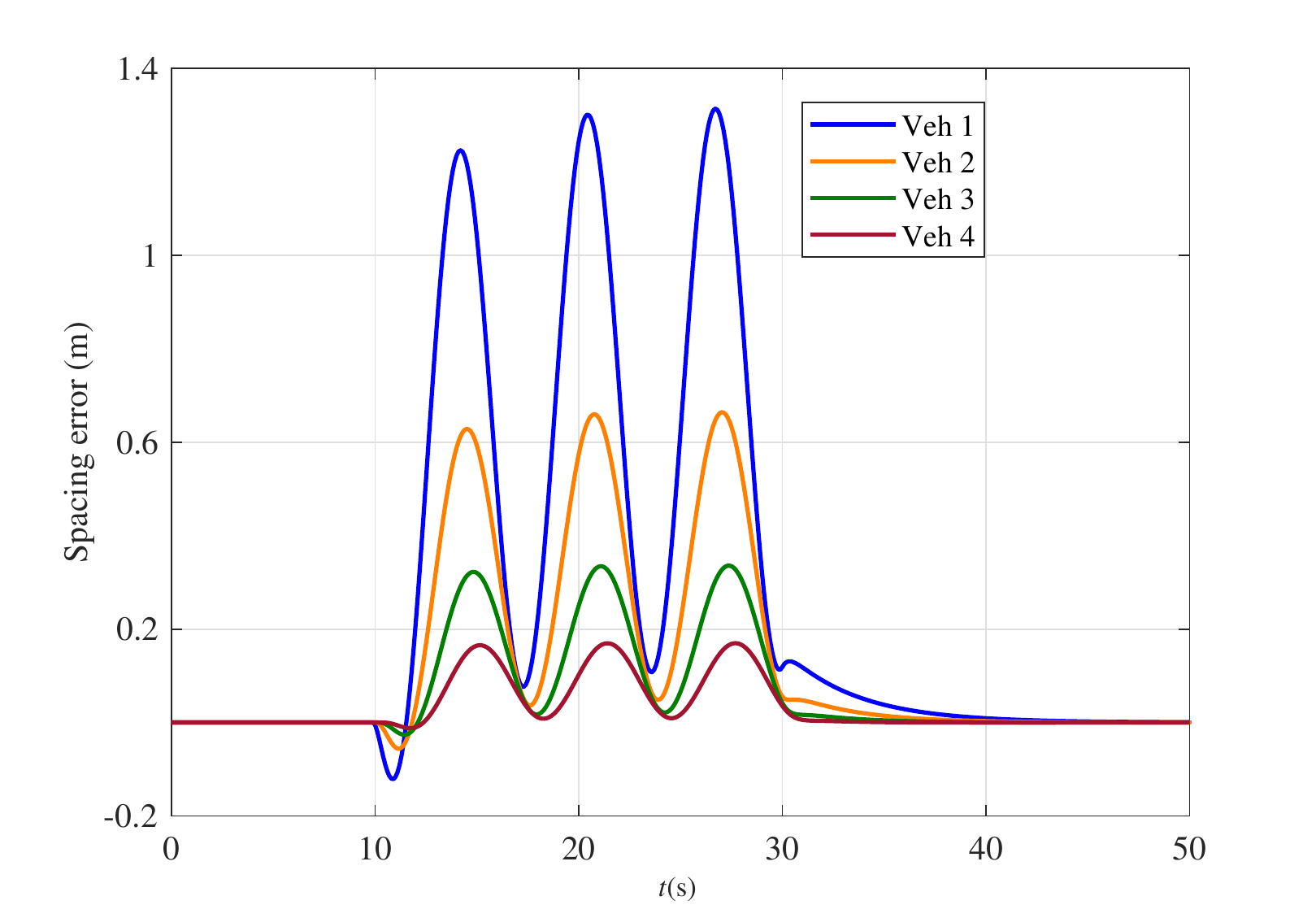}
      \label{fig1b}
    }
 \caption{Effects of different control gains.}
 \label{control_gains_fig}
\end{figure}

This subsection shows the effects of choosing different control gains. The leader vehicle's acceleration varies as $\dot{v}_0(t)=-sin\left(t\right)$ at $10\leq t\leq30$ (s), $M=4$, $\tau=0.1$s and $h=0.2$s. The control gain vectors in Fig. \ref{control_gains_fig}(a) and Fig. \ref{control_gains_fig}(b) are given by $\left[K_v, K_{v_o }, K_x, K_{x_o}\right]=\left[1.5,1.5,0.273,0.281\right]$ and $\left[K_v, K_{v_o }, K_x, K_{x_o}\right]=\left[1.5,1.5,0.4,0.4\right]$, respectively, which are chosen from the the feasible regions in Section III.

It is seen in Fig. \ref{control_gains_fig}(a) and Fig. \ref{control_gains_fig}(b) that both control gain vectors can achieve plant stability and string stability, since they belong to the feasible regions mentioned in Section III. Although the platoon experiences the same external disturbance, the slightly different values of the control gains may cause significantly different performance behaviors, i.e., the control gains used in Fig. \ref{control_gains_fig}(a) make the follower vehicles' space errors vary more drastically, and the platoon needs to spend more time on reaching the stability, compared to the case of control gains used in Fig. \ref{control_gains_fig}(b).

\subsection{Effects of External Disturbance}
 \begin{figure}
     \centering
    \subfigure[]{
         \centering
         \includegraphics[width=3.4 in,height=2.5 in]{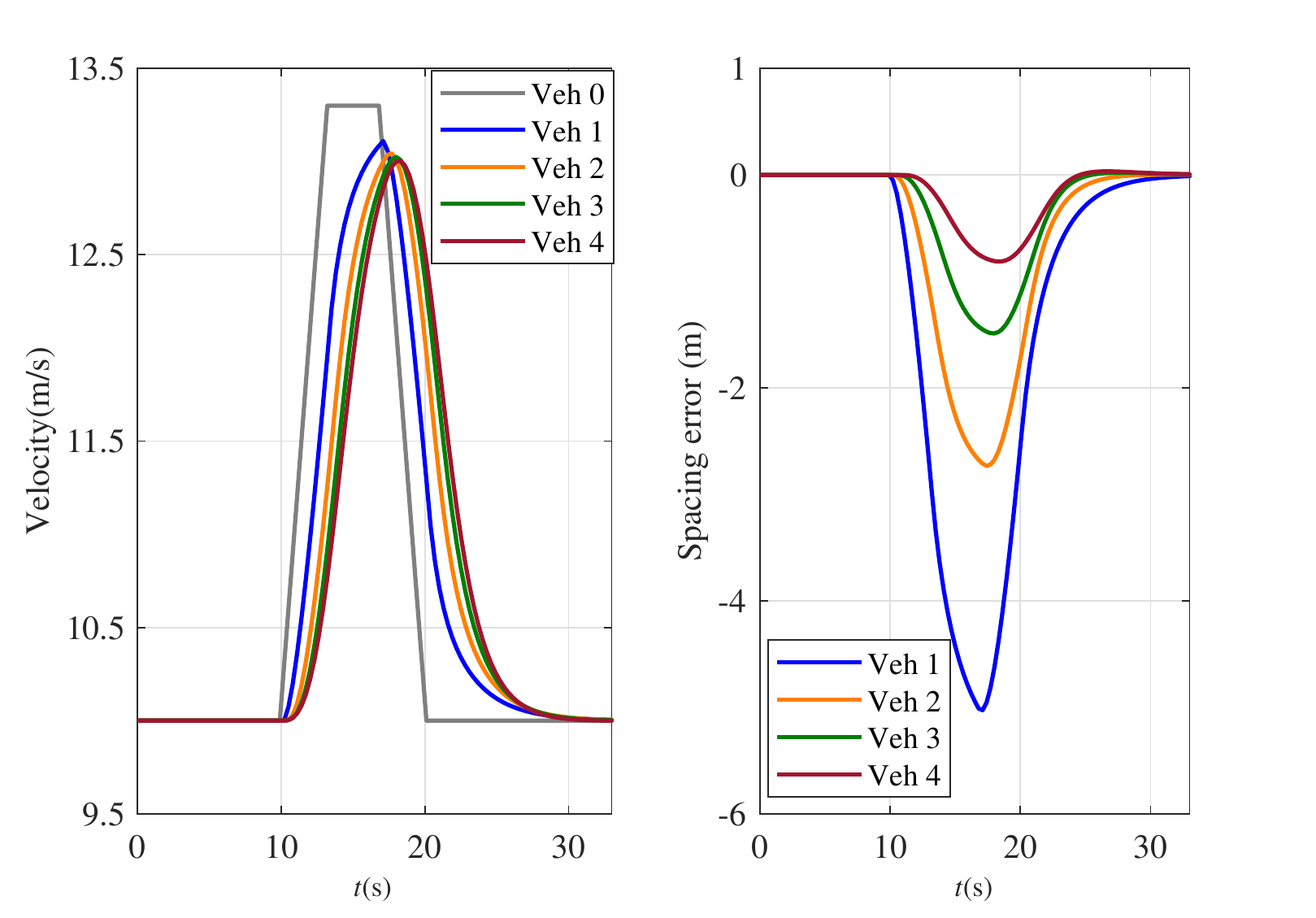}
      \label{fig1a}
     }
     \subfigure[]{
         \centering
         \includegraphics[width=3.4 in,height=2.5 in]{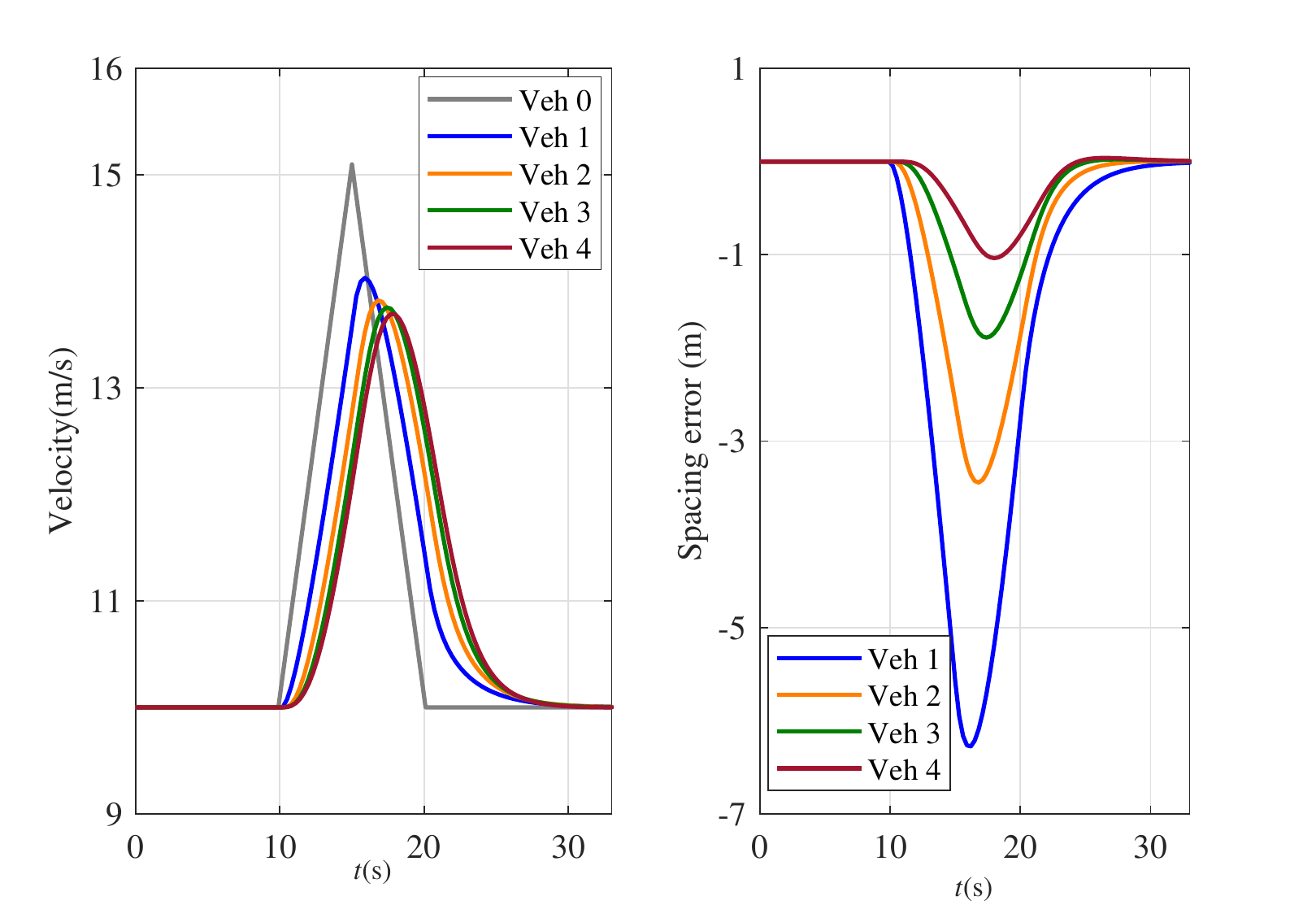}
      \label{fig1b}
    }
 \caption{Effects of different external disturbances.}
 \label{disturbance_fig}
\end{figure}
This subsection shows the effects of different external disturbances imposed on the leader vehicle. Specifically, the external disturbances of the platoon for the simulations in Fig. \ref{disturbance_fig}(a) and Fig. \ref{disturbance_fig}(b) are given by
 \begin{align}\label{smu_disturbzero}
 \dot{v}_0(t)=
\left\{ \begin{array}{l}
1,\quad \,\,10 \le t \le 13 \mathrm{s},\quad  \\
0,\quad \,\,13 < t \le 17 \mathrm{s},\quad  \\
- 1,\quad 17 < t \le 20 \mathrm{s},\quad  \\
 \end{array} \right.
\end{align}
and
\begin{align}\label{smu_disturb}
 \dot{v}_0(t)=
\left\{ \begin{array}{l}
1,\quad \,\,10 \le t \le 15 \mathrm{s},\quad  \\
- 1,\quad 15 < t \le 20 \mathrm{s},\quad  \\
 \end{array} \right.
\end{align}
respectively. The other basic simulation parameters are identical in the results of Fig. \ref{disturbance_fig}(a) and Fig. \ref{disturbance_fig}(b), namely the control gain vector $\left[K_v, K_{v_o }, K_x, K_{x_o}\right]=\left[0.75,0.75,0.249,0.228\right]$, $M=4$, $\tau=0.3$s and $h=0.2$s.

It is seen from Fig. \ref{disturbance_fig}(a) and Fig. \ref{disturbance_fig}(b) that although the platoon stability for these two types of external disturbances are achieved at almost the same time, the external disturbance given by \eqref{smu_disturb} forces the follower vehicles to change their moving speeds more rapidly and results in larger spacing errors, compared to the type of external disturbance given by \eqref{smu_disturbzero}. Such dramatic changes of the vehicles' driving states during the external disturbance may need to be properly addressed in practice, due to the fact that different vehicles may have velocity limitations under hardware constraints.

\begin{figure}
     \centering
    \subfigure[]{
         \centering
         \includegraphics[width=3.4 in,height=2.5 in]{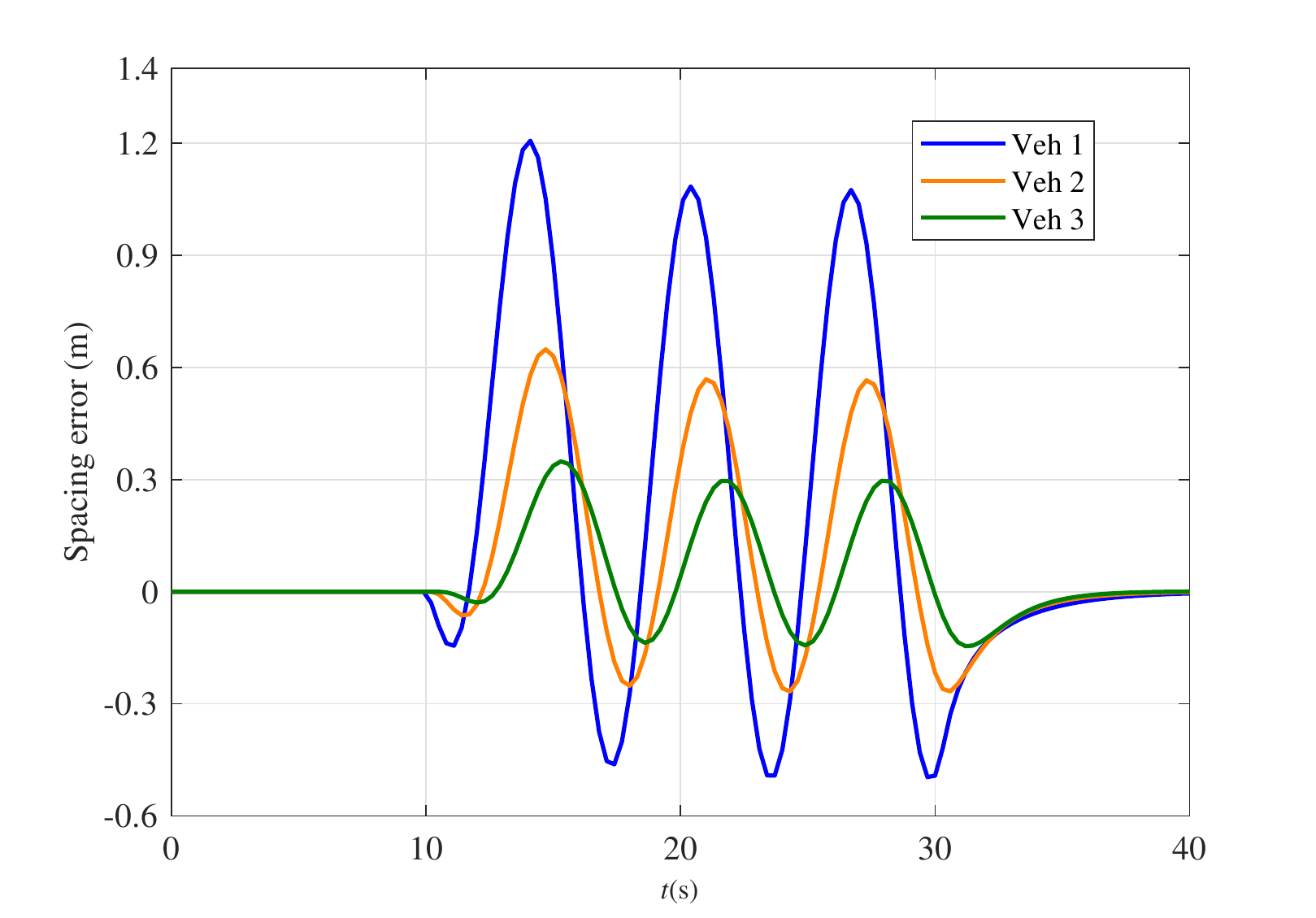}
      \label{fig1a}
     }
     \subfigure[]{
         \centering
         \includegraphics[width=3.4 in,height=2.5 in]{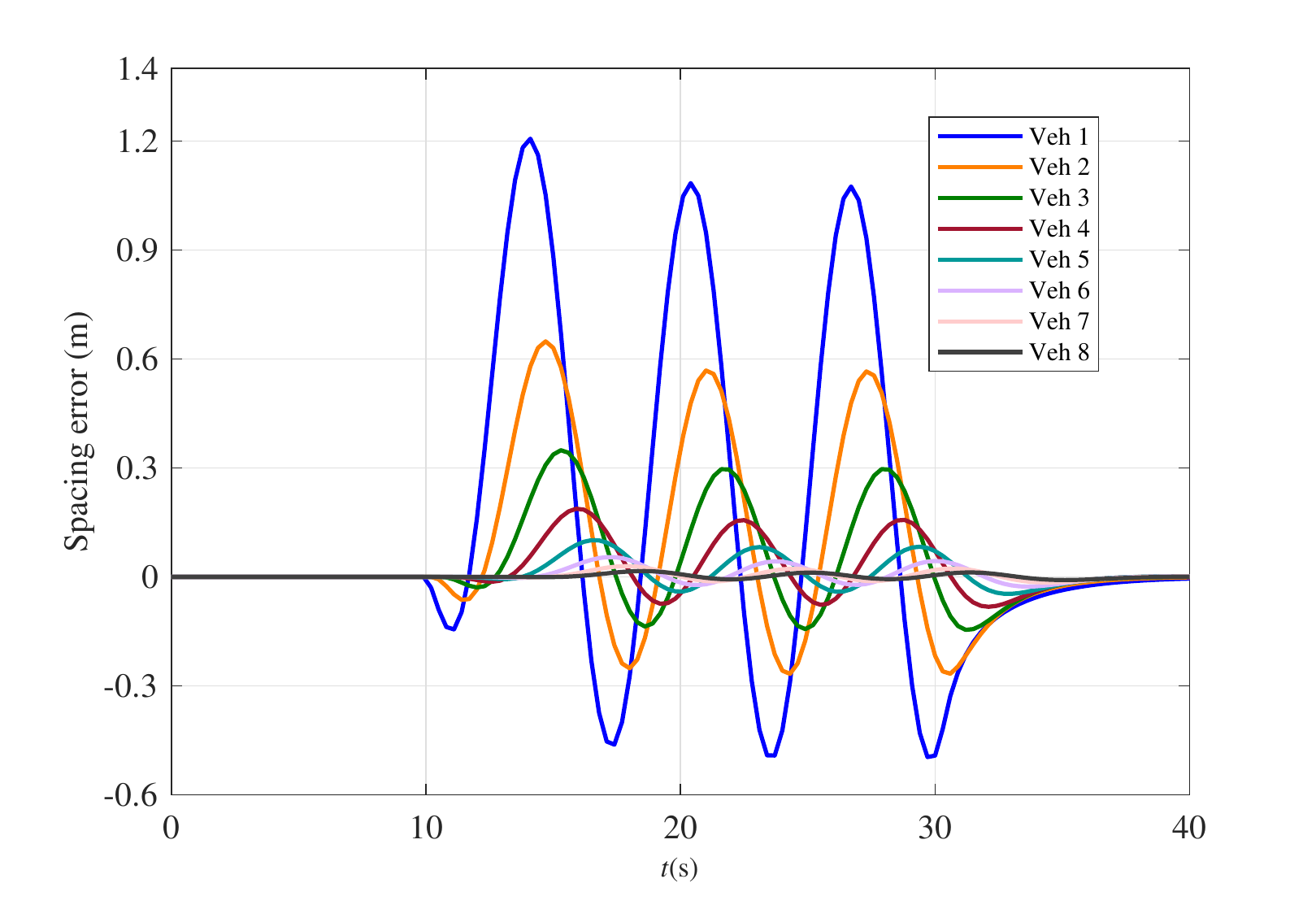}
      \label{fig1b}
    }
 \caption{Effects of different platoon sizes.}
 \label{platoonsize_fig}
\end{figure}

\subsection{Effects of Platoon Size}
This subsection shows the effects of platoon size. In the simulations, there are two platoons consisting of three and eight follower vehicles, respectively, the control gain vector $\left[K_v, K_{v_o }, K_x, K_{x_o}\right]=\left[0.75,0.75,0.249,0.228\right]$, the leader vehicle's acceleration varies as $\dot{v}_0(t)=-sin\left(t\right)$ at $10\leq t\leq30$ (s), $\tau=0.3$s and $h=0.2$s.

It is seen from Fig. \ref{platoonsize_fig}(a) and Fig. \ref{platoonsize_fig}(b) that when the control gains and other system parameters are fixed, changing the platoon size has negligible effect on the spacing errors of the follower vehicles, which confirms the scalability of the proposed platooning design. In addition, platoons with different sizes have nearly same disturbance time period before reaching the system stability.

\subsection{Effects of Delay}
This subsection shows the effects of delay. In the simulations, we consider two delay cases, i.e., $\tau=0.1$s in Fig. \ref{delay_fig}(a) and $\tau=0.3$s in Fig. \ref{delay_fig}(b),  the platoon consists of six follower vehicles besides the leader vehicle,  the control gain vector $\left[K_v, K_{v_o }, K_x, K_{x_o}\right]=\left[0.75,0.75,0.249,0.228\right]$, the leader vehicle's acceleration varies as $\dot{v}_0(t)=-sin\left(t\right)$ at $10\leq t\leq30$ (s),  and $h=0.2$s.
\begin{figure}
     \centering
    \subfigure[]{
         \centering
         \includegraphics[width=3.4 in,height=2.5 in]{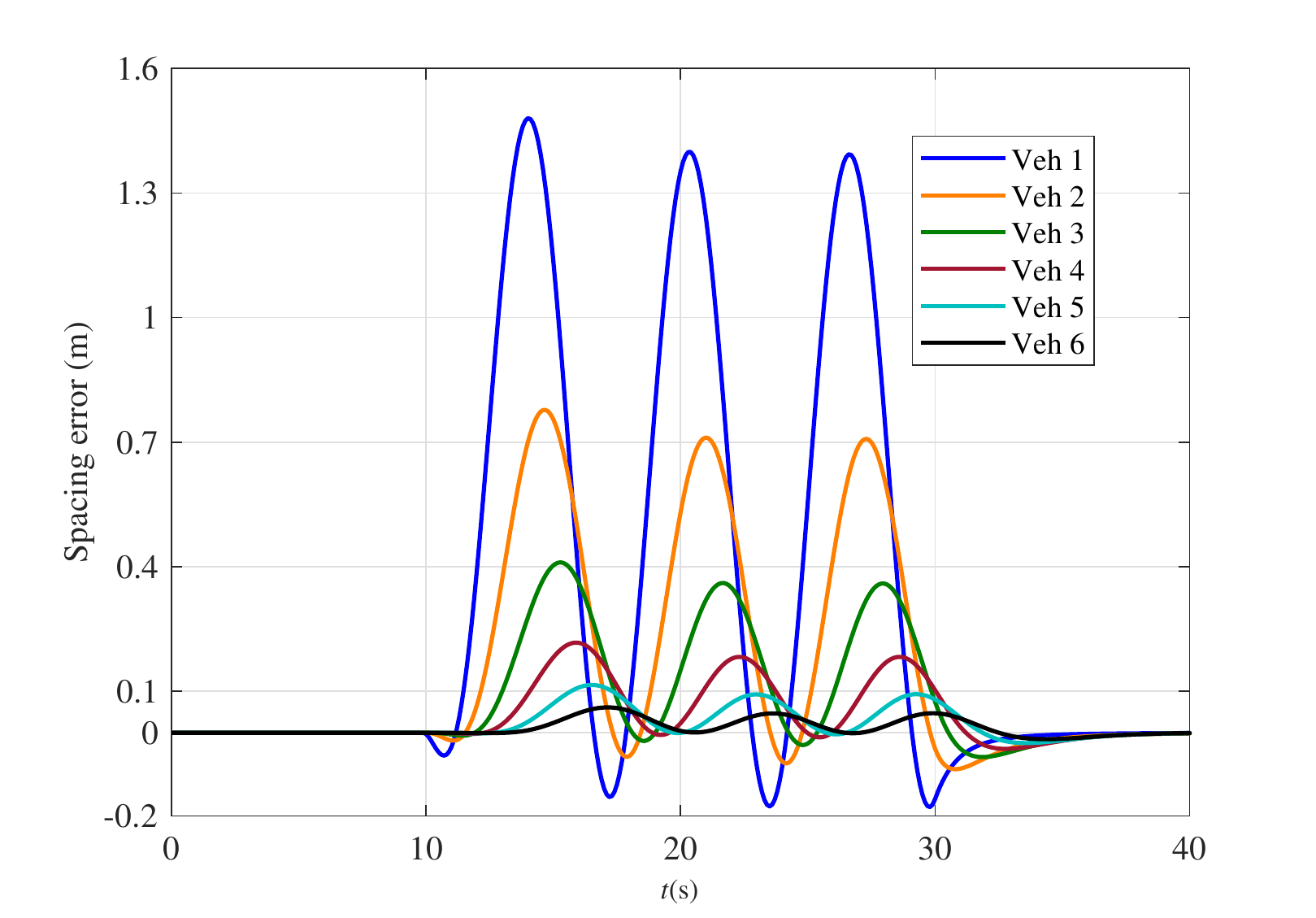}
      \label{fig1a}
     }
     \subfigure[]{
         \centering
         \includegraphics[width=3.4 in,height=2.5 in]{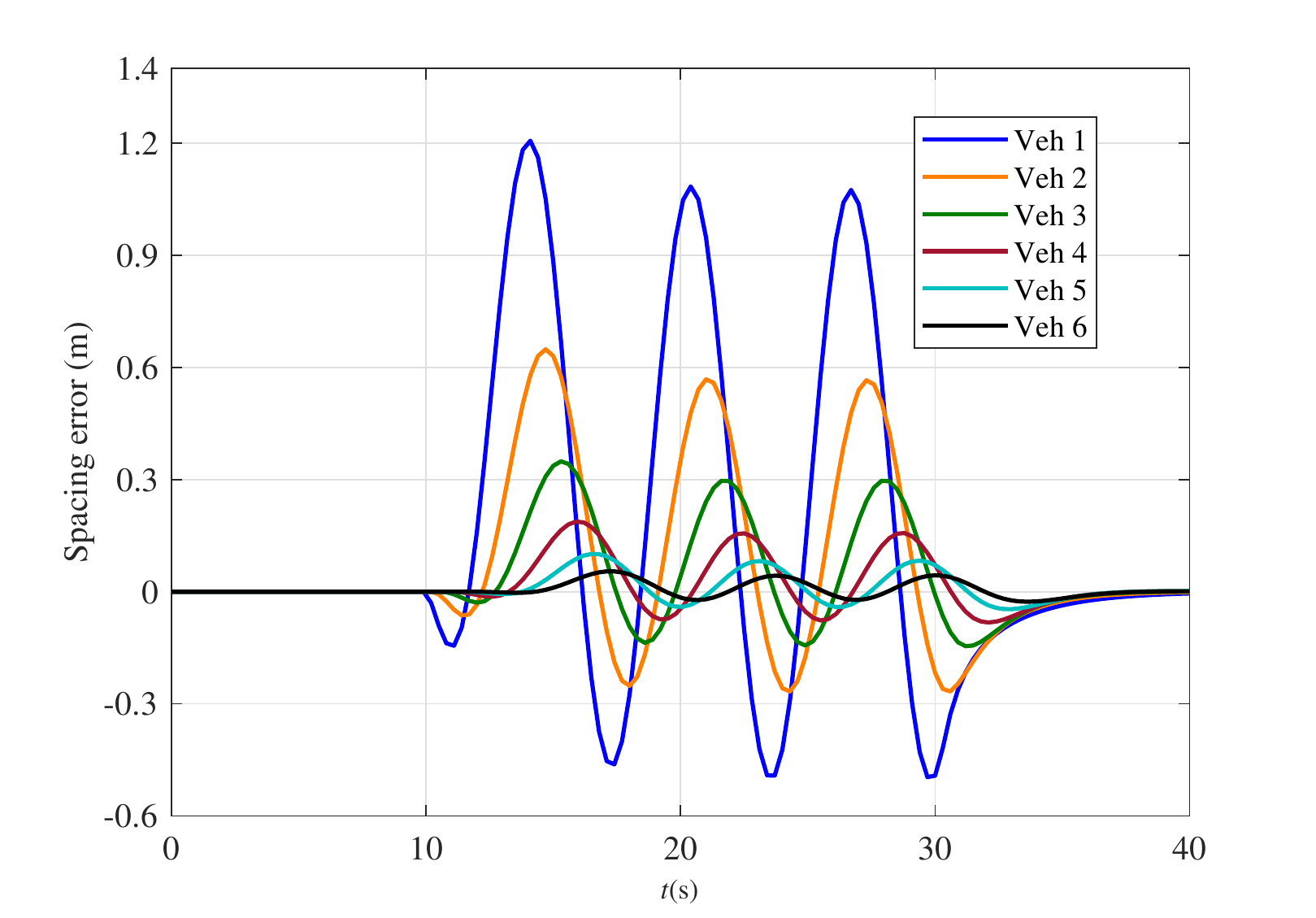}
      \label{fig1b}
    }
 \caption{Effects of different levels of delay.}
 \label{delay_fig}
\end{figure}

It is seen from Fig. \ref{delay_fig}(a) and Fig. \ref{delay_fig}(b) that when the control gains and other system parameters are fixed, different levels of delay have a big impact on the spacing errors during the disturbance time period. However, the time of reaching the system stability is nearly unaltered for different delay cases. Through the comparison with the results in Fig. \ref{platoonsize_fig}, it is again confirmed that platoons with different sizes has negligible effect on the stability and efficiency of the proposed design when the rest of system parameters and external disturbance are identical.

\section{Conclusions}\label{conclusion_section}
This paper concentrated on the V2I-based platooning systems, where RSUs have the capabilities of massive MIMO and edge computing. By considering the effect of delay,  an efficient platooning control approach was developed. We demonstrated that the proposed platooning design can achieve both plant stability and string stability by selecting control gains in the derived feasible regions. Moreover, we provided a tractable method to explicitly quantify the relationships between the platoon's velocity, platoon size/length, radio resources and handover. By using our derivations, the platoon's velocity, radio sources and handover can be easily determined. Simulation results confirmed the efficiency of the proposed platooning design, and the effects of different control gains, external disturbances, platoon sizes and delays on the performance were comprehensively illustrated.

\section*{Appendix A: Proof of Theorem 1}
\label{App:theo_1}
\renewcommand{\theequation}{A.\arabic{equation}}
\setcounter{equation}{0}
The necessary and sufficient condition for plant stability is given via the D-subdivision approach~\cite{D_subdivision_method}. Let $s_0=\xi + j w$, the characteristic equation $\Theta\left(s_0\right) = 0$ can be decomposed into real and imaginary parts, which are
\begin{align}
&\mathrm{Re}: \eta \xi \cos \left( {\tau w} \right) + \eta w\sin \left( {\tau w} \right) + \lambda \cos \left( {\tau w} \right) = e^{\tau \xi } \left( w^2-\xi ^2\right),
\label{Re_Im_1}\\
&\mathrm{Im}: \eta w\cos \left( {\tau w} \right) - \eta \xi \sin \left( {\tau w} \right) - \lambda \sin \left( {\tau w} \right) + 2e^{\tau \xi } \xi w = 0. \label{Re_Im_2}
\end{align}
By letting $\xi=0$, the D-curves can be expressed as
\begin{align}\label{Re_Im}
&\mathrm{Re}: \eta w\sin \left( {\tau w} \right) + \lambda \cos \left( {\tau w} \right) = w^2,  \\
&\mathrm{Im}: \eta w\cos \left( {\tau w} \right) = \lambda \sin \left( {\tau w} \right).
\end{align}
The above equation can be equivalently written as
\begin{align}\label{eta_lamda}
\lambda = w^2 \cos \left( {\tau w} \right), \;\; \eta  = w \sin \left( {\tau w} \right).
\end{align}
Note that $w>0$ and $\tau w \in \left(2k\pi, \frac{\pi }{{2 }} + 2k\pi \right), k=0,1,2,\cdots,$ since $\lambda >0$ and $\eta > 0$.

To determine the crossing direction from stability to instability along the D-curves, we first take the first-order derivative of \eqref{Re_Im_1} and \eqref{Re_Im_2} with respect to $\eta$
at $\xi=0$ (along the D-curves), after mathematical manipulations, the first-order derivative of $\xi$ at $\xi=0$ is
\begin{align}\label{eta_xi_derivative}
\frac{{d\xi }}{{d\eta }} = \frac{\frac{w^2 \left( {\tau \eta  - 2\cos \left( {\tau w} \right)} \right)}{\left( {2w + \tau \lambda \sin \left( {\tau w} \right) - \tau \eta w\cos \left( {\tau w} \right) - \eta \sin \left( {\tau w} \right)} \right)^2}}{{\frac{\left( {\tau \eta w\sin \left( {\tau w} \right) + \tau \lambda \cos \left( {\tau w} \right) - \eta \cos \left( {\tau w} \right)} \right)^2}{\left( {2w + \tau \lambda \sin \left( {\tau w} \right) - \tau \eta w\cos \left( {\tau w} \right) - \eta \sin \left( {\tau w} \right)} \right)^2}  + 1 }}.
\end{align}
From \eqref{eta_xi_derivative}, we see that when $\eta > \frac{2\cos \left( {\tau w} \right)}{\tau}$, $\frac{{d\xi }}{{d\eta }} > 0$, i.e., there exists the positive real part of the characteristic root, and the plant stability is violated as $\eta$ increases. Based on \eqref{eta_lamda}, $\eta = \frac{\pi }{{2\tau }} + \frac{{2k\pi }}{\tau }\left(k=0,1,2,\cdots\right)$ as $\lambda =0$.
Considering the fact that $\frac{{d\xi }}{{d\eta }} > 0$ with $\eta = \frac{\pi }{{2\tau }}$, $\left(\lambda, \eta\right) =\left(0, \frac{\pi }{{2\tau }}\right)$ is a corner point of the stability region, which means that $\tau w \in \left(0, \frac{\pi }{{2 }}\right)$.

Likewise, taking the first-order derivative of \eqref{Re_Im_1} and \eqref{Re_Im_2} with respect to $\lambda$
at $\xi=0$ (along the D-curves), after mathematical manipulations, we have
\begin{align}\label{lambda_xi_derivative}
\frac{{d\xi }}{{d\lambda }} = \frac{{\frac{{\tau \lambda {\rm{ + }}\eta }}{{\left( {\tau \lambda \cos \left( {\tau w} \right) + \tau \eta w\sin \left( {\tau w} \right) - \eta \cos \left( {\tau w} \right)} \right)^2 }}}}{{\frac{{\left( {2w + \tau \lambda \sin \left( {\tau w} \right) - \tau \eta w\cos \left( {\tau w} \right) - \eta \sin \left( {\tau w} \right)} \right)^2 }}{{\left( {\tau \lambda \cos \left( {\tau w} \right) + \tau \eta w\sin \left( {\tau w} \right) - \eta \cos \left( {\tau w} \right)} \right)^2 }} + 1}}.
\end{align}
From \eqref{lambda_xi_derivative}, we see that $\frac{{d\xi }}{{d\lambda }} > 0$ for arbitrary $\lambda$ value, which means that the plant stability is violated as $\lambda$ increases. Thus, we can finally obtain the feasible region $\mathcal{G}\left(\tau\right)$ given by \eqref{w_eta_gamma_plane}.

\section*{Appendix B: Proof of Theorem 2}
\label{App:theo_2}
\renewcommand{\theequation}{B.\arabic{equation}}
\setcounter{equation}{0}
String stability is achieved when $\left|\mathcal{H}_i\left(jw\right)\right| < 1$. Based on \eqref{error_transfer_eq}, $\left|\mathcal{H}_i\left(jw\right)\right|$ is given by
\begin{align}\label{B_11}
\left|\mathcal{H}_i\left(jw\right)\right| =
\sqrt {\frac{{K_v^2 w^2  + K_x^2 }}{\Xi \left( w \right)+ K_v^2 w^2  + K_x^2}},
\end{align}
where
\begin{align}\label{B_12}
\Xi \left( w \right) & =
w^4  - 2\eta\sin \left( {\tau w} \right)w^3 \nonumber\\
&~~ + \left( K_x^2 h^2  + 2K_x \left( {K_v  + K_{v_o } } \right)h + K_{v_o }^2  + 2K_v K_{v_o } \right)w^2 \nonumber\\
&~~~- 2\lambda \cos \left( {\tau w} \right) w^2 + K_{x_o }^2  + 2K_x K_{x_o }.
\end{align}
Note that both the numerator and denominator of \eqref{B_11} have the positive term $K_v^2 w^2  + K_x^2$, thus $\left|\mathcal{H}_i\left(jw\right)\right| < 1$ is equivalently transformed as $\Xi \left( w \right)> 0$, $\forall w \geq 0$. Considering the fact that $\sin \left( {\tau w} \right) \le \tau w$ and $\cos \left( {\tau w} \right) \leq 1$, we have
 \begin{align}\label{B_13}
- 2\eta \sin \left( {\tau w} \right)w^3  \ge  - 2\eta \tau w^4, \;\; 2\lambda \cos \left( {\tau w} \right)w^2 \leq 2\lambda w^2 .
\end{align}
Based on  \eqref{B_12} and \eqref{B_13}, the following inequality is obtained as
 \begin{align}\label{B_14}
\Xi \left( w \right) \geq &\left( {1 - 2\eta \tau } \right)w^4  + K_x^2 h^2 w^2 \nonumber\\
&~~+ \left(2K_x \left( {K_v  + K_{v_o } } \right)h + K_{v_o }^2 \right) w^2  \nonumber\\
&~~~~~+ 2\left( {K_v K_{v_o }  - \lambda } \right)w^2  + K_{x_o }^2  + 2K_x ,
\end{align}
When $\eta \le \frac{1}{{2\tau }}$ and $\lambda \leq K_v K_{v_o }$, the right-hand-side of the inequality is positive, thus $\Xi \left( w \right) > 0$, and complete the proof.

\bibliographystyle{IEEEtran}

\begin{thebibliography}{10}
\providecommand{\url}[1]{#1}
\csname url@samestyle\endcsname
\providecommand{\newblock}{\relax}
\providecommand{\bibinfo}[2]{#2}
\providecommand{\BIBentrySTDinterwordspacing}{\spaceskip=0pt\relax}
\providecommand{\BIBentryALTinterwordstretchfactor}{4}
\providecommand{\BIBentryALTinterwordspacing}{\spaceskip=\fontdimen2\font plus
\BIBentryALTinterwordstretchfactor\fontdimen3\font minus
  \fontdimen4\font\relax}
\providecommand{\BIBforeignlanguage}[2]{{%
\expandafter\ifx\csname l@#1\endcsname\relax
\typeout{** WARNING: IEEEtran.bst: No hyphenation pattern has been}%
\typeout{** loaded for the language `#1'. Using the pattern for}%
\typeout{** the default language instead.}%
\else
\language=\csname l@#1\endcsname
\fi
#2}}
\providecommand{\BIBdecl}{\relax}
\BIBdecl

\bibitem{C_Y_Liang_1999}
C.-Y. Liang and H. Peng, ``Optimal adaptive cruise control with guaranteed
  string stability,'' \emph{Veh. Syst. Dynamics}, vol. 32, no. 4-5, pp.
  313-330, 1999.

\bibitem{Ramzi2003}
R.~Abou-Jaoude, ``{ACC} radar sensor technology, test requirements, and test
  solutions,'' \emph{{IEEE} Trans. Intell. Transp. Syst.}, vol.~4, no.~3, pp.
  115--122, 2003.

\bibitem{Daniel_Work2020}
G.~Gunter, C.~Janssen, W.~Barbour, R.~E. Stern, and D.~B. Work, ``Model-based
  string stability of adaptive cruise control systems using field data,''
  \emph{{IEEE} Trans. Intell. Veh.}, vol.~5, no.~1, pp. 90--99, Mar. 2020.

\bibitem{X_Liu_2001}
{X. Liu}, A.~{Goldsmith}, S.~S. {Mahal}, and J.~K. {Hedrick}, ``Effects of
  communication delay on string stability in vehicle platoons,'' in \emph{IEEE
  Intell. Transp. Syst. Conf. (ITSC)}, 2001, pp. 625--630.

\bibitem{Soncu_2014}
S.~\"{O}nc\"{u}, J.~{Ploeg}, N.~{van de Wouw}, and H.~{Nijmeijer},
  ``Cooperative adaptive cruise control: Network-aware analysis of string
  stability,'' \emph{IEEE Trans. Intell. Transp. Syst.}, vol.~15, no.~4, pp.
  1527--1537, Aug. 2014.

\bibitem{G_Naus_2010}
G.~J.~L. {Naus}, R.~P.~A. {Vugts}, J.~{Ploeg}, M.~J.~G. {van de Molengraft},
  and M.~{Steinbuch}, ``String-stable {CACC} design and experimental
  validation: {A} frequency-domain approach,'' \emph{{IEEE} Trans. Veh.
  Technol.}, vol.~59, no.~9, pp. 4268--4279, Nov. 2010.

\bibitem{J_Ploeg_2011}
J.~{Ploeg}, B.~T.~M. {Scheepers}, E.~{van Nunen}, N.~{van de Wouw}, and
  H.~{Nijmeijer}, ``Design and experimental evaluation of cooperative adaptive
  cruise control,'' in \emph{IEEE Conf. on Intell. Transp. Syst. (ITSC)}, 2011,
  pp. 260--265.

\bibitem{Vicente2014}
V.~Milan\'{e}s, S.~E. Shladover, J.~Spring, C.~Nowakowski, H.~Kawazoe, and
  M.~Nakamura, ``Cooperative adaptive cruise control in real traffic
  situations,'' \emph{{IEEE} Trans. Intell. Transp. Syst.}, vol.~15, no.~1, pp.
  296--305, Feb. 2014.

\bibitem{3GPP_TS_V2X}
3GPP TS 22.186, ``Enhancement of 3GPP support for V2X scenarios; Stage 1
  (Release 16),'' June 2019.

\bibitem{Shaw2007}
E.~Shaw and J.~K. Hedrick, ``String stability analysis for heterogeneous
  vehicle strings,'' in \emph{Proc. American Control Conf.}, 2007, pp.
  3118--3125.

\bibitem{YuYu_Lin_2017}
Y.-Y. Lin and I.~Rubin, ``Integrated message dissemination and traffic
  regulation for autonomous {VANETs},'' \emph{{IEEE} Trans. Veh. Technol.},
  vol.~66, no.~10, pp. 8644--8658, Oct. 2017.

\bibitem{B_Liu_2017}
B.~{Liu}, D.~{Jia}, K.~{Lu}, D.~{Ngoduy}, J.~{Wang}, and L.~{Wu}, ``A joint
  control-communication design for reliable vehicle platooning in hybrid
  traffic,'' \emph{{IEEE} Trans. Veh. Technol.}, vol.~66, no.~10, pp.
  9394--9409, Oct. 2017.

\bibitem{Shengbo_ITSC_2018}
Y.~Bian, Y.~Zheng, S.~E. Li, Z.~Wang, Q.~Xu, J.~Wang, and K.~Li, ``Reducing
  time headway for platoons of connected vehicles via multiple-predecessor
  following,'' in \emph{IEEE Conf. on Intell. Transp. Syst. (ITSC)}, 2018, pp.
  1240--1245.

\bibitem{S_Darbha_2019}
S.~Darbha, S.~Konduri, and P.~R. Pagilla, ``Benefits of {V2V} communication for
  autonomous and connected vehicles,'' \emph{{IEEE} Trans. Intell. Transp.
  Syst.}, vol.~20, no.~5, pp. 1954--1963, May 2019.

\bibitem{Yuanheng_Zhu_2019}
Y.~Zhu, D.~Zhao, and Z.~Zhong, ``Adaptive optimal control of heterogeneous
  {CACC} system with uncertain dynamics,'' \emph{{IEEE} Control Syst.
  Technol.}, vol.~27, no.~4, pp. 1772--1779, July 2019.

\bibitem{Tengchan_2019}
T.~Zeng, O.~Semiari, W.~Saad, and M.~Bennis, ``Joint communication and control
  for wireless autonomous vehicular platoon systems,'' \emph{{IEEE} Trans.
  Commun.}, vol.~67, no.~11, pp. 7907--7922, Nov. 2019.

\bibitem{Michal_Sybis_2019}
M.~Sybis, V.~Vukadinovic, M.~Rodziewicz, P.~Sroka, A.~Langowski, K.~Lenarska,
  and K.~Weso{\l}owski, ``Communication aspects of a modified cooperative
  adaptive cruise control algorithm,'' \emph{{IEEE} Trans. Intell. Transp.
  Syst.}, vol.~20, no.~12, pp. 4513--4523, Dec. 2019.

\bibitem{Linjun_2016}
L.~Zhang and G.~Orosz, ``Motif-based design for connected vehicle systems in
  presence of heterogeneous connectivity structures and time delays,''
  \emph{{IEEE} Trans. Intell. Transp. Syst.}, vol.~17, no.~6, pp. 1638--1651,
  June 2016.

\bibitem{Shengbo_2019_mag}
Y.~Zheng, Y.~Bian, S.~Li, and S.~E. Li, ``Cooperative control of heterogeneous
  connected vehicles with directed acyclic interactions,'' \emph{{IEEE} Intell.
  Transp. Syst. Mag.}, pp. 1--17, 2019.

\bibitem{Mario_2015}
M.~di~Bernardo, A.~Salvi, and S.~Santini, ``Distributed consensus strategy for
  platooning of vehicles in the presence of time-varying heterogeneous
  communication delays,'' \emph{{IEEE} Trans. Intell. Transp. Syst.}, vol.~16,
  no.~1, pp. 102--112, Feb. 2015.

\bibitem{HaitaoXing_2020}
H.~Xing, J.~Ploeg, and H.~Nijmeijer, ``Compensation of communication delays in
  a cooperative {ACC} system,'' \emph{{IEEE} Trans. Veh. Technol.}, vol.~69,
  no.~2, pp. 1177--1189, Feb. 2020.

\bibitem{A_He_2017}
A.~He, L.~Wang, Y.~Chen, K.~K. Wong, and M.~Elkashlan, ``Spectral and energy
  efficiency of uplink {D2D} underlaid massive {MIMO} cellular networks,''
  \emph{{IEEE} Trans. Commun.}, vol.~65, no.~9, pp. 3780--3793, Sept. 2017.

\bibitem{Vicente_2012}
V.~Milan\'{e}s, J.~Villagr\'{a}, J.~Godoy, J.~Sim\'{o}, J.~P{\'{e}}rez, and
  E.~Onieva, ``An intelligent {V2I}-based traffic management system,''
  \emph{{IEEE} Trans. Intell. Transp. Syst.}, vol.~13, no.~1, pp. 49--58, Mar.
  2012.

\bibitem{LeiChen_2016}
L.~Chen and C.~Englund, ``Cooperative intersection management: {A} survey,''
  \emph{{IEEE} Trans. Intell. Transp. Syst.}, vol.~17, no.~2, pp. 570--586,
  Feb. 2016.

\bibitem{yuyu_lin_ITA}
Y.-Y. Lin and I.~Rubin, ``Infrastructure aided networking and traffic
  management for autonomous transportation,'' in \emph{IEEE Inf. Theory Appl.
  Workshop (ITA)}, 2017, pp. 1--7.

\bibitem{Montanaro_2018}
U.~{Montanaro}, S.~{Fallah}, M.~{Dianati}, D.~{Oxtoby}, T.~{Mizutani}, and
  A.~{Mouzakitis}, ``On a fully self-organizing vehicle platooning supported by
  cloud computing,'' in \emph{Fifth Int. Conf. Internet of Things: Syst.,
  Management and Security}, 2018, pp. 295--302.

\bibitem{Vilalta_2019}
F.~V\'{a}zquez-Gallego, R.~Vilalta, A.~Garc\'{\i}a, F.~Mira, S.~{V\'{\i}a},
  R.~{Mu\~{n}oz}, J.~Alonso-Zarate, and M.~Catalan-Cid, ``Demo: {A} mobile edge
  computing-based collision avoidance system for future vehicular networks,''
  in \emph{IEEE INFOCOM WKSHPS}, 2019, pp. 904--905.

\bibitem{Chang_2020}
B.-J. Chang and J.-M. Chiou, ``Cloud computing-based analyses to predict
  vehicle driving shockwave for active safe driving in intelligent
  transportation system,'' \emph{{IEEE} Trans. Intell. Transp. Syst.}, vol.~21,
  no.~2, pp. 852--866, Feb. 2020.

\bibitem{xiaoyanhu_2020}
X.~Hu, L.~Wang, K.~K. Wong, M.~Tao, Y.~Zhang, and Z.~Zheng, ``Edge and central
  cloud computing: {A} perfect pairing for high energy efficiency and
  low-latency,'' \emph{{IEEE} Trans. Wireless Commun.}, vol.~19, no.~2, pp.
  1070--1083, Feb. 2020.

\bibitem{LX_2011}
L.~Xiao and F.~Gao, ``Practical string stability of platoon of adaptive cruise
  control vehicles,'' \emph{{IEEE} Trans. Intell. Transp. Syst.}, vol.~12,
  no.~4, pp. 1184--1194, Dec. 2011.

\bibitem{del_Peral-Rosado}
J.~A. {del Peral-Rosado}, M.~A. {Barreto-Arboleda}, F.~{Zanier},
  G.~{Seco-Granados}, and J.~A. {L¨®pez-Salcedo}, ``Performance limits of {V2I}
  ranging localization with {LTE} networks,'' in \emph{14th Workshop on
  Positioning, Navigation and Commun. (WPNC)}, 2017, pp. 1--5.

\bibitem{Wymeersch_2017}
H.~{Wymeersch}, G.~{Seco-Granados}, G.~{Destino}, D.~{Dardari}, and
  F.~{Tufvesson}, ``{5G} mmwave positioning for vehicular networks,''
  \emph{IEEE Wireless Commun.}, vol.~24, no.~6, pp. 80--86, Dec. 2017.

\bibitem{Jianglin2020}
J.~Lan and D.~Zhao, ``Min-max model predictive vehicle platooning with
  communication delay,'' \emph{{IEEE} Trans. Veh. Technol.}, vol.~69, no.~11,
  pp. 12\,570--12\,584, Dec. 2020.

\bibitem{H_Xing_2016}
H. Xing, J. Ploeg, and H. Nijmeijer, ``Pad\'{e} approximation of delays in
  cooperative {ACC} based on string stability requirements,'' \emph{IEEE Trans.
  Intell. Veh.}, vol. 1, no. 3, pp. 277-286, Sep. 2016.

\bibitem{D_subdivision_method}
T. Insperger and G. St\'{e}p\'{a}n, \emph{Semi-Discretization for Time-Delay
  Systems}. New York, NY, USA: Springer-Verlag, 2011.

\bibitem{Marzetta_2010_Nonc}
T.~L. Marzetta, ``Noncooperative cellular wireless with unlimited numbers of
  base station antennas,'' \emph{{IEEE} Trans. Wireless Commun.}, vol.~9,
  no.~11, pp. 3590--3600, Nov. 2010.

\bibitem{Bjorson_mag_2016}
E.~Bj{\"{o}}rnson, E.~G. Larsson, and T.~L. Marzetta, ``Massive {MIMO:} ten
  myths and one critical question,'' \emph{IEEE Commun. Mag.}, vol.~54, no.~2,
  pp. 114--123, Feb. 2016.

\bibitem{ngo2013energy}
H.~Q. Ngo, E.~G. Larsson, and T.~L. Marzetta, ``Energy and spectral efficiency
  of very large multiuser {MIMO} systems,'' \emph{{IEEE} Trans. Commun.},
  vol.~61, no.~4, pp. 1436--1449, Apr. 2013.

\bibitem{TS37_340}
3GPP TS 37.340, ``Multi-connectivity; Stage 2 (Release 16),'' Mar. 2020.

\bibitem{lifeng_mag2018}
L. Wang, K. K. Wong, S. Jin, G. Zheng, and R. W. Heath Jr., ``A new look at
  physical layer security, caching, and wireless energy harvesting for
  heterogeneous ultra-dense networks,'' \emph{IEEE Commun. Mag.}, vol. 56, no.
  6, pp. 49-55, June 2018.

\end{thebibliography}

\end{document}